\documentclass[leqno]{article}

\usepackage{amsfonts,amsmath,amssymb,amsthm,graphicx,psfrag,subfigure,enumerate}
\usepackage[round]{natbib}

\usepackage[dvips]{hyperref}
\usepackage{color}

\psfrag{zero}{\tiny $0$}
\psfrag{one}{\tiny $1$}
\psfrag{two}{\tiny $2$}
\psfrag{three}{\tiny $3$}
\psfrag{four}{\tiny $4$}
\psfrag{Ak}{\footnotesize $A(K)$}
\psfrag{Ek}{\footnotesize $E(K)$}
\psfrag{Ertk}{\footnotesize $E(e^{rT}K)$}
\psfrag{KS0}{\footnotesize $K-S_0$}
\psfrag{ErtkS0}{\footnotesize $e^{rT}K-S_0$}
\psfrag{Akt}{$\tilde{A}(K)$}

\psfrag{0}{\tiny $0$}
\psfrag{1}{\tiny $1$}
\psfrag{2}{\tiny $2$}
\psfrag{3}{\tiny $3$}
\psfrag{4}{\tiny $4$}
\psfrag{5}{\tiny $5$}
\psfrag{6}{\tiny $6$}
\psfrag{7}{\tiny $7$}
\psfrag{0.2}{\tiny $0.2$}
\psfrag{0.4}{\tiny $0.4$}
\psfrag{0.6}{\tiny $0.6$}
\psfrag{0.8}{\tiny $0.8$}
\psfrag{Ak}{\scriptsize $A(K)$}
\psfrag{Ek}{\scriptsize $E(K)$}
\psfrag{Ertk}{\scriptsize $E(e^{rT}K)$}
\psfrag{Erttk}{\scriptsize $E(e^{r(T-t^{*})}K)$}
\psfrag{KS0}{\scriptsize $K-S_0$}
\psfrag{AktEk}{\scriptsize $A(K) = E(K)$}
\psfrag{ErtkS0}{\scriptsize $e^{-rT}K-S_0$}
\psfrag{Akt}{\scriptsize $\tilde{A}(K)$}
\psfrag{S_t}{\scriptsize $S_t$}
\psfrag{Ks}{}

\newtheorem{thm}{Theorem}[section]
\newtheorem{lemma}[thm]{Lemma}

\newtheorem{corollary}[thm]{Corollary}
\newtheorem{prop}[thm]{Proposition}

\newtheorem{remark}[thm]{Remark}
\newtheorem*{rmk}{Remark}

\title{Model-independent no-arbitrage conditions on American put options}
\author{A. M. G. Cox\footnote{email: a.m.g.cox@bath.ac.uk, web: http://www.maths.bath.ac.uk/$\sim$mapamgc/} \, and Christoph Hoeggerl\footnote{email: ch487@bath.ac.uk}\\ Department of Mathematical Sciences,\\ University of Bath, Bath, U.K.}

\begin{document}
\maketitle

\begin{abstract}
\noindent We consider the pricing of American put options in a model-independent setting: that is, we do not assume that asset prices behave according to a given model, but aim to draw conclusions that hold in any model. We incorporate market information by supposing that the prices of European options are known. 

In this setting, we are able to provide conditions on the American Put prices which are necessary for the absence of arbitrage. Moreover, if we further assume that there are finitely many European and American options traded, then we are able to show that these conditions are also sufficient. To show sufficiency, we construct a model under which both American and European options are correctly priced at all strikes simultaneously. In particular, we need to carefully consider the optimal stopping strategy in the construction of our process.\\
\end{abstract}

\hspace*{19 mm} Legendre-Fenchel transformation, Skorokhod embedding

\newpage

\section{Introduction}

The standard approach to pricing contingent claims is to postulate a model and to determine the prices as the discounted expected payoffs under some equivalent risk-neutral measure.  A major problem with this approach is that no model can capture the real world behaviour of asset prices fully and this leaves us prone to model risk. An alternative to the model-based approach is to try to ask: {\it when are observed prices consistent with some model?} It can often then be shown that when there is no model which is consistent with observed prices, then there exists an arbitrage which works under all models. Since these properties hold independently of any model, we shall refer to such notions as being model-independent.

The basis of the model-independent approach, which we follow and 
which can be traced back to the insights of \citet{BLZ}, is to suppose European call options are sufficiently liquidly traded that they are no longer considered as being priced under a model, but are obtained exogenously from the market. According to \citet{BLZ} call prices for a fixed maturity date $T$ can then be used to recover the marginal distribution of the underlying at time $T$. This way contingent claims depending only on the distribution at the fixed time $T$ can be priced without having made any assumptions on the underlying model. \citet{HOB} first observed that, by considering the possible martingales which are consistent with the inferred law, one can often infer extremal properties of the class of possible price processes, and then use this to deduce bounds on the prices of other options on the same underlying when using the European option prices as hedging instrument. This approach has been extended in recent years to pricing various path-dependent options using Skorokhod embedding techniques. \citet{HOB}, for example, determined how to hedge lookback options. \citet{BHR} showed how to hedge barrier options. \citet{DAH} determined the range of traded option prices for European calls, whereas \citet{CO2,CO1} found robust prices on double touch and no-touch barrier options, and \citet{CAW} have extended results of \citet{DUP3} and \citet{CAL} regarding options on variance. We refer to \citet{HO1} for an overview of this literature. Recently, \citet{GLT} applied the Kantorovich duality to transform the problem of superhedging under volatility uncertainty to an optimal transportation problem, where they managed to recover the results from \citet{HOB} for lookback options.

In this paper, we will be interested in the prices of American put options, and in particular, whether a given set of American Put prices and co-terminal European Put prices are consistent with the absence of model-independent arbitrage. Our only financial assumptions are that we can buy and sell both types of derivatives initially at the given prices, and that we can trade in the underlying frictionlessly  at a discrete number of times. Under these conditions, we are able to give a set of simple conditions on the prices which, if violated, guarantee the existence of an arbitrage under {\it any} model for the asset prices. In addition, we show that these conditions are sufficient in the restricted setting where only finitely many European and American options trade. Specifically, given prices which satisfy our conditions, we are able to produce a model and a pricing measure that reproduce these prices. Clearly, the restriction to a finite number of traded options is not a significant restriction for practical purposes.

Several authors have considered arbitrage conditions on American options in the model-independent setting. Closely related to our work is the work of \citet{EKH}, who determine a time-homogeneous stock price process consistent with given perpetual option prices, and the subsequent generalisation to a wider class of optimal stopping problems by \citet{HAK}, however both these papers work under the assumption that the price process lies in the class of time-homogenous diffusions, an assumption that we do not make. Also of relevance is a working paper of \citet{NEU}, who found arbitrage bounds for a single American option with a finite horizon through a linear programming approach. Neuberger takes as given the prices of European options at all maturities, rather than a single maturity as we do, and is able to relate the range of arbitrage-free prices to solutions of a linear programming problem. Although we only consider prices with a single common maturity date, the conclusions we provide are more concrete. Finally, \citet{SHA} has obtained an upper and lower bound on an American put option with fixed strike from given American put options with the same maturity, but different strikes. He does not consider the impact of co-terminal European options, and his resulting conditions are therefore easily shown to be satisfied by some model in a one-step procedure.


The main results in this paper therefore concern necessary and sufficient conditions for the absence of arbitrage in quoted co-terminal European and American options: specifically, we are able to give four conditions which we show to be necessary and sufficient. It is well known (e.g. \citet{DAH} or  \citet{CAM}) which conditions must be placed on European Put options for the absence of model-independent arbitrage, so we are interested only in conditions on the American options in terms of the European prices. Three of the conditions are not too surprising: there are known upper and lower bounds, and the American prices must be increasing and convex. However we also establish a fourth condition in terms of the value and the gradient of the European and American options, which we have not found elsewhere in the literature. This condition also has a natural representation in terms of the Legendre-Fenchel transform.

To establish that our conditions are necessary for the absence of model-independent arbitrage, we show that there exists a simple strategy that creates an arbitrage should any of the conditions be violated. It turns out to be much harder to show that our conditions are sufficient: to do this, it is necessary for us to specialise to the case where there are only finitely many traded options, and in this setting, we are able to construct a model under which all options are correctly priced. This requires us to construct both a price process, and keep track of the value function of an optimal stopping problem. The description of this process will comprise a large amount of the content of this paper. While this approach is in spirit close to many of the papers which exploit Skorokhod embedding technologies (e.g.~\citet{CO2,CO1,CAW,HOB}), 
there are also a number of differences: specifically, that we do not use a time-changed Brownian motion, nor do we attempt to construct an `extremal' embedding; rather, the embedding step will form a fairly small part of the description of our overall construction.

The construction of the process which attains a given set of prices is described by means of an algorithm: from a set of possible American and European Put prices, we shall describe how the prices may be `split' into two new pairs of functions, which can then be considered as independent sets of European and American prices at a later time. By repeated splitting, we are able to show that the problem eventually reduces to a trivial model which we can describe easily. From this recursive procedure, we are able to reconstruct a process which satisfies all our required conditions. It will turn out that the price process we recover is fairly simple: the price will grow at the interest rate until a non-random time, at which the price jumps to one of two fixed levels. This splitting continues until the maturity date, when it jumps to a final position.

The conditions that we derive should be of interest both for theoretical and practical purposes. They are important for market makers and speculators alike, as a violation of the conditions represents a clear misspecification in the prices under any model, allowing for arbitrage which can be realised using a simple semi-static trading strategy. Our conditions also present simple consistency checks that can be applied to verify that the output of any numerical procedure is valid, and to extrapolate prices which are not quoted from existing market data. In addition, the results we present can also be used as a mechanism to provide an estimate of model-risk associated with a particular position in a set of American options.

The rest of the paper is organised as follows. In Section~\ref{snc} we discuss the necessary conditions and show that a violation of any of these conditions leads to model-independent arbitrage. In Section~\ref{ssc} we will then argue that for any given set of prices $A$ and $E$ that satisfy the necessary conditions there exists a model and a plausible price process, hence the conditions also have to be sufficient for the absence of model-independent arbitrage. The Appendix contains some additional proofs that would have only impaired the reading fluency of the paper.

\section{Necessary conditions on the American Put price function A}\label{snc} 

Assume we are given an underlying asset $S$ which does not pay dividends and which may be traded frictionlessly. In addition, we may hold cash which accrues interest at a constant rate $r > 0$. Furthermore, we will be able to trade options on the underlying at given prices at time 0 only, and these options will always have a common maturity date $T$.

As we are interested in model-independent behaviour we do not begin by specifying a model or probability measure. It is therefore not immediately clear what arbitrage or the absence of arbitrage means. Along the lines of \citet{DAH}, we say that there exists model-independent arbitrage if we can construct a semi-static portfolio in the underlying and the options that has strictly negative initial value and only non-negative subsequent cashflows. Further we consider a portfolio to be semi-static if it involves holding a position in the options and the underlying, where the position in the options was fixed at the initial time and the position in the underlying can only be altered finitely many times by a self-financing strategy. 

There are situations where no model-independent arbitrage opportunities exist, but where we still can find a semi-static portfolio such that the initial portfolio value is non-positive, all subsequent cashflows are non-negative and the probability of a positive cashflow is non-zero, if only the null sets of the underlying model are known. These trading strategies were termed weak arbitrage in  \citet{DAH}.

We will consider two cases, one where we are given European put option prices at a finite number of strikes and one where we are given a European price function $E$ for all strikes $K\geq 0$. When there are only finitely many option prices given we shall assume that the European Call prices satisfy the conditions given in Theorem 3.1 of \citet{DAH} --- that is, that there is neither a model-independent, nor a weak arbitrage --- and that we have Put-Call parity, then we can immediately derive conditions on the European put price function $E$ that have to be satisfied for the absence of model-independent arbitrage:

\begin{lemma}\label{coe} 
Suppose the prices of European Put options with maturity $T$ are given for a finite number of strikes $K_1,...,K_n$. Denote the European Put option prices as a function of the strike $K$ by $E$, where $E$ is interpolated linearly between the given prices. Then the European Put prices are free of model-independent and weak arbitrage opportunities if and only if the following conditions are satisfied:
\begin{enumerate}
\item The European Put price function $E$ is increasing and convex in K with $E(0) = 0$.
\item The function $(e^{-rT}K-S_0)_{+}$ is a lower bound for $E$.
\item The function $e^{-rT}K$ is an upper bound for $E$.
\item For any $K \geq 0$ with $E(K) > e^{-rT}K-S_0$ we have $E'(K+) < e^{-rT}$.
\end{enumerate} 
Here $S_0$ is the current price of the underlying asset.
\end{lemma}

In the situation where European put prices are given for all positive strikes we can replace the fourth condition of Lemma~\ref{coe} by $|E(K)-(e^{-rT}K-S_0)| \rightarrow 0$ as $K \rightarrow \infty$ under the assumption that there is no {\it weak free lunch without vanishing risk} (for details see \citet{CO1}).

Returning to the situation where there are finitely many strikes given we can conclude due to \citet{BLZ} that these conditions are sufficient to imply the existence of a probability measure $\mu$ on $\mathbb{R}^{+}$ such that $E(K) = \int (e^{-rT}K- x)_{+} \mu(dx)$. In addition the following result has to hold.

\begin{lemma}\label{blr} 
If there exists a probability measure $\mu$ on $\mathbb{R}^{+}$ such that $\int x \mu(dx) = S_0$ and $E(K) = \int (e^{-rT}K-x)_{+} \mu(dx)$, then the European put price function $E$ satisfies the conditions of Lemma~\ref{coe}.
\end{lemma}

\begin{proof}
The first condition follows from the fact that $\mu$ is a probability measure and that the integrand $(e^{-rT}K-x)_{+}$ of $E$ is positive and convex. The lower bound is obtained by applying Jensen's inequality to the convex function $x \mapsto (e^{-rT}K-x)_{+}$, whereas the upper bound follows from $(e^{-rT}K-x)_{+} \leq e^{-rT}K$ as $\mu$ is only defined on $\mathbb{R}^{+}$.

In the case of the fourth condition we will prove the contrapositive. Note that $E'(K) = e^{-rT}\int \mathbf{1}_{[0,e^{-rT}K]}(x) \mu(dx)$. Since $\mu$ is a probability measure and we assume that there exists a $K^{*}$ with  $E'(K^{*}) \geq e^{-rT}$ we can conclude that $\mu([0,e^{-rT}K^{*}]) = 1$, hence for any $K \geq K^{*}$ we must have $$E(K) = \int(e^{-rT}K-x)\mu(dx) = e^{-rT}K - S_0,$$
which completes the proof.
\end{proof}

Note further that if we are given the finite set of prices $E(K_1),...,E(K_n)$ such that $E(K_i) > e^{-rT}K_i - S_0$ for all $i=1,...,n$, then we can always extend the set of strikes by a final strike $K_{n+1} $ for which we set $E(K_{n+1}) = e^{-rT}K_{n+1} - S_0$. The new set of prices then satisfies again the conditions given in Lemma~\ref{coe} as long as $K_{n+1}$ was chosen such that $$K_{n+1} > \frac{(E(K_n)+S_0)K_{n-1}-(E(K_{n-1})+S_0)K_n}{E(K_n)-E(K_{n-1})-e^{-rT}K_n+e^{-rT}K_{n-1}},$$ where the term on the right hand-side is the strike where the linear piece $\frac{E(K_n)-E(K_{n-1})}{K_n-K_{n-1}}(K-K_{n-1}) + E(K_{n-1})$ intersects with $e^{-rT}K-S_0$. We can therefore assume that we are always given a set of prices where the last one lies on the lower bound $e^{-rT}K-S_0$. 


Under these assumptions we are able to state the main result of this section, Theorem~\ref{nco}, which will give us conditions on $A$ that necessarily have to be fulfilled for $A$ to be an arbitrage-free American Put price function, assuming we are given the prices of European Put options satisfying the conditions above.

\begin{thm}\label{nco}
If $A$ is an arbitrage-free American Put price function then it must satisfy the following conditions:
\begin{enumerate}[(i)]
\item The American Put price function $A$ is increasing and convex in $K$.
\item For any $K \geq 0$ we have
\begin{align*}
A'(K+)K-A(K) \geq E'(K+)K-E(K).
\end{align*}
\item The function $\max\{E(K),K-S_0\}$ is a lower bound for $A$.
\item The function $E(e^{rT}K)$ is an upper bound for $A$.
\end{enumerate}
\end{thm}

With the exception of {\it (ii)}, these properties are not too surprising: it is well known that the American Put price must be convex and increasing, and it is also clear that the price of the American option must dominate both the corresponding European option, and its immediate exercise value. The upper bound given in {\it (iv)} appears to date back to \citet{MAR}. Although he works in the Black-Scholes setting, his arguments hold also in the general case under consideration here. 

\begin{remark} 
\begin{enumerate}[(i)]
\item
Recall that the Legendre-Fenchel transform of a function $f:\mathbb{R} \rightarrow \mathbb{R}$ is given by $f^{*}(k)=\sup_{x \in \mathbb{R}}\{kx - f(x)\}$, so we can rewrite the second condition of Theorem~\ref{nco} as 
\begin{align} 
A^{*}(A'(K+)) \geq E^{*}(E'(K+)) 
\end{align} 
for all $K \geq 0$. This can be seen by rewriting $f^{*}(k) = -\inf_{x \in \mathbb{R}} \{ f(x) - kx\}$ and noting that the function $f$ is given for $x \geq 0$, and is non-negative, increasing and convex in our case. 

\item
It follows directly from condition $(ii)$ of Theorem \ref{nco} that the early exercise premium $A-E$ has to be increasing, as $A'(K) - E'(K) \geq \frac{A(K)-E(K)}{K}$ is positive. However, these statements are not equivalent, and there exist examples where the early-exercise premium is increasing, and the other necessary conditions are satisfied, but condition {\it (ii)} of the theorem fails.
\end{enumerate}
\end{remark}

\begin{proof}[Proof of Theorem~\ref{nco}]
We will prove each statement separately using model-independent arbitrage arguments.
To see that the American Put price function $A$ has to be increasing in the strike $K$ we will assume the contrary so that we have $A(K_1) > A(K_2)$ for  any two positive strikes $K_1 < K_2$. 
We can then make an initial profit of $A(K_1) - A(K_2)$ by short selling an American put option with strike $K_1$ and buying an American put option with strike $K_2$. To guarantee that any subsequent cashflow is positive we only have to close out the long position when the American with strike $K_1$ is exercised, leaving us with $K_2 - K_1 > 0$. We can then conclude that the function $A(K)$ has to be increasing in $K$, since there would be an arbitrage opportunity otherwise.

As in the case before we will prove that the function $A$ has to be convex by assuming that $\alpha A(K_1) + (1-\alpha)A(K_2) < A(\alpha K_1 + (1-\alpha)K_2)$ for some $\alpha \in [0,1]$ and $K_1 < K_2$ holds. This way a portfolio consisting of a short position in an American put option with strike $\alpha K_1 + (1-\alpha)K_2$ and a long position of $\alpha$ units in an American put option with strike $K_1$ and $(1-\alpha)$ units in an American put option with strike $K_2$ has strictly negative initial cost. If we close out the long positions when the counterparty in the short contract exercises we have at the time of exercise, denoted $\tau^{*}$, at least
\begin{align*}
\alpha(K_1 - S_{\tau^*}) + (1-\alpha)(K_2 - S_{\tau^*}) + (S_{\tau^*} - (\alpha K_1 + (1-\alpha) K_2)) = 0.
\end{align*}
Therefore absence of arbitrage implies that $A(K)$ has to be convex in $K$.

As proved in Lemma \ref{cae} we have that the condition in $(ii)$ is equivalent to
\begin{align}\label{lca} 
\frac{1}{\epsilon} (A(K+\epsilon)-A(K)) - \frac{1}{K}A(K) \geq \frac{1}{\epsilon}(E(K+\epsilon)-E(K)) - \frac{1}{K} E(K)
\end{align}
for all $K \geq 0$ and any $\epsilon$ with $0 < \epsilon \leq \tilde{\epsilon}(K)$. Suppose the condition in (\ref{lca}) is violated, then we can make an initial profit by selling $\frac{1}{\epsilon}$ units of $E(K+\epsilon)$ and $\frac{K+\epsilon}{K\epsilon}$ units of $A(K)$, while buying $\frac{1}{\epsilon}$ units of $A(K+\epsilon)$ and $\frac{K+\epsilon}{K\epsilon}$ units of $E(K)$.

Suppose now that the shorted American was exercised at time $\tau$, where we then also exercised the long American to obtain at maturity $T$ a cashflow of
\begin{multline*}
\frac{1}{\epsilon}\left[(e^{r(T-\tau)}(K+\epsilon)-S_T)-(K+\epsilon - S_T)_{+}\right] \\
- \frac{K+\epsilon}{K\epsilon}\left[(e^{r(T-\tau)}K-S_T) -(K-S_T)_{+} \right],
\end{multline*}
which is equal to
\begin{align*}
\begin{cases}
\frac{1}{K}S_T &,S_T \geq K+\epsilon\\
\frac{K+\epsilon}{K\epsilon}(S_T - K) &,S_T \in [K, K+\epsilon]\\
0 &,S_T \leq K,
\end{cases}
\end{align*}
implying arbitrage.

To obtain the upper bound we suppose $E(e^{rT}K) < A(K)$. We sell the American option with strike $K$, and buy the European with strike $e^{rT}K$, making an initial profit of $A(K) - E(e^{rT}K)$. At maturity the American generates a cashflow $S_T - Ke^{r(T-\tau)}$, depending on the time $\tau$ when the American option was exercised. Further we receive from the European option the amount $(e^{rT}K - S_T)_{+}$. In the case where $K > e^{-rT}S_T$ we have
\begin{align*}
(e^{rT}K - S_T) + (S_T - Ke^{r(T- \tau)}) = e^{rT}K(1-e^{-r\tau}) > 0.
\end{align*}
Whereas for $K < e^{-rT}S_T$ the European Put $E(e^{rT}K)$ has $0$ payoff, but by the assumption on $K$ the American Put now gives us
\begin{align*}
S_T - Ke^{r(T- \tau)} > 0.
\end{align*}
Analogously we can show that the lower bound has to hold and we have therefore proved all the statements of the theorem.
\end{proof}

We observe that the arbitrage strategies in the proof above involve exercising long positions in an American at a time determined by the holder of a different option. Initially, this might seem worrying, as we may be exercising suboptimally. However, we note that in the arbitrages constructed above, whenever we exercise early we also hold a short position in an American with smaller strike, which is exercised at the same time. In the following lemma we show that if this option is exercised optimally, then it is optimal for us also to exercise. We note that a similar result has been used in a similar context by \citet{SHA}.

\begin{lemma}\label{oeas} 
Suppose prices are consistent with some model. If for a positive strike $K_1$ it is optimal to exercise the American put option immediately, then for any strike $K_2 > K_1$ it is also optimal to exercise the American put option with strike $K_2$ immediately.
\end{lemma}

\begin{proof}
Assume that it is optimal to exercise the American put option with strike $K_1$ immediately, then, since prices are consistent with some model, there exists a (risk-neutral) pricing measure and we can write
\begin{align*}
A(K_1) = \sup_{0 \leq \tau \leq T}\mathbb{E}\left[e^{-r\tau}(K_1 - S_{\tau})_+\right] = K_1 - S_0.
\end{align*}
Suppose now that there exists a $K_2 > K_1$ for which it is not optimal to exercise the American Put immediately. Assuming that the optimal stopping time exists we denote it by $\tau^* >0$ and have
\begin{align*}
A(K_2) &= \sup_{0 \leq \tau \leq T}\mathbb{E} \left[ e^{-r\tau}(K_2 - S_{\tau})_+ \right]\\
&= e^{-r\tau^*}\mathbb{E}(K_2 - S_{\tau^*})_+ \\
&= e^{-r\tau^*}\mathbb{E}(K_2 - K_1 + K_1 - S_{\tau^*})_+\\
&\leq e^{-r\tau^*}(K_2 - K_1) + e^{-r\tau^*}\mathbb{E}(K_1 - S_{\tau^*})_+\\
&\leq K_2 - S_0.
\end{align*}
This is a contradiction to the assumption that it is not optimal for $K_2$ to exercise immediately and hence we showed that it is also optimal to exercise the American option for all $K_2 > K_1$ immediately. Note that in the case where the optimal stopping time does not exist a similar argument can be used.
\end{proof}

\begin{remark} The upper and lower bounds on the American Put price, given in $(iii)$ and $(iv)$ of Theorem~\ref{nco} respectively, can also be seen to be tight, that is, there exist models that attain the bounds as American Put price function. In the case of the lower bound the following underlying price process satisfies $A(K) = \max\{(K-S_0)_{+},E(K)\}$. Set
\begin{align*}
S_t =
\begin{cases}
e^{-r(T-t)} \mathbb{E}Y  & ,\, t \in [0,T) \\
Y & ,\, t=T
\end{cases}
\end{align*}
where $Y$ is an integrable random variable with distribution $\mu$.  This process grows at the interest rate up to $T$, where it jumps to its final distribution $Y$. The discounted price process $e^{-rt}S_t$ is by definition a martingale with respect to its natural filtration $\mathcal{F}_t^{S}$. As the process grows at the interest rate between the times $0$ and maturity $T$ we know that the payoff obtained by exercising immediately will always exceed the payoff for any time $t \in (0,T)$, hence the only possible stopping times are $0$ and $T$, which gives $A(K) = \max\{(K-S_0)_{+},E(K)\}$. 

In the case of the upper bound the following price process $(S_t)_{t \geq 0}$ has $E(e^{rT}K)$ as American price function. Set
\begin{align*}
S_t =
\begin{cases}
e^{-rT} \mathbb{E}Y  & ,\, t = 0 \\
e^{-r(T-t)}Y & ,\, t \in (0,T]
\end{cases}
\end{align*}
where again $r > 0$ is the interest rate, $T$ the maturity date and $Y$ the integrable final distribution. 
It is straightforward to check that this is indeed a martingale, and if we consider the sequence of stopping times $\tau_n = \frac{1}{n}$, we get

\begin{align*}
A(K) 
&\ge \lim_{n \rightarrow \infty} e^{-rT}\mathbb{E}\left[(e^{rT - r/n}K-Y)_+\right]\\
&= e^{-rT}\mathbb{E}\left[(e^{rT}K - S_T)_+\right]\\
&= E(e^{rT}K),
\end{align*}
as required.
\end{remark}

\section{Sufficiency of the conditions on the American Put price function A}\label{ssc} 

In order to show that the necessary conditions in Theorem~\ref{nco} are also sufficient for the absence of model-independent arbitrage it is enough to determine for any given set of American and European Put prices a market model such that the European and American put option prices satisfy $e^{-rT}\mathbb{E}(K-S_T)_{+} = E(K)$ and $\sup_{0 \leq \tau \leq T} \mathbb{E}(K-S_{\tau})_{+} = A(K)$, respectively. A market model consists of a filtered probability space $(\Omega,(\mathcal{F}_t)_{0 \leq t \leq T},\mathbb{P})$ and an underlying price process $(S_t)_{0 \leq t \leq T}$ where $(e^{-rt}S_t)_{t \geq 0}$ is a martingale under $\mathbb{P}$.

In general it appears to be a harder task to show that the conditions of Theorem~\ref{nco} are also sufficient, particularly if it is assumed that a continuum of option strikes trade. Consequently, we shall consider a slightly restricted setup (although one that is still practically very relevant): henceforth we will assume that we are given American and European prices for a finite number of strikes, from which we will extrapolate general functions $A$ and $E$ for which the conditions of Theorem~\ref{nco} and Lemma~\ref{coe} hold.\footnote{Given a finite set of traded options which are derived from some model, it is not the case that their linear interpolation will satisfy the conditions of Theorem~\ref{nco} and Lemma~\ref{coe} automatically, however it seems plausible that there should be some larger set of strikes which do. Indeed, we believe that, given a set of traded option prices, either we can construct a piecewise linear extension satisfying the conditions of Theorem~\ref{nco} or there exists model-independent arbitrage. However, this is a non-trivial result and we leave a formal proof to subsequent work.}


Contrary to the embedding problem considered in \citet{BUE}, \citet{COU} or \citet{DAH}, where marginals for multiple fixed times are given, the definition of the American put option requires us to incorporate the American prices into $(S_t)_{t\geq0}$ at the unknown optimal stopping time $\tau^{*}$ before the European prices are embedded at maturity  $T$.

Suppose the piecewise linear functions $A$ and $E$ satisfy the conditions of Theorem~\ref{nco} and Lemma \ref{coe} and are given as follows. For each strike $K \geq 0$ the function $E$ is the European Put price corresponding to the discrete marginal distribution $\mu = p_1 \delta_{K^E_1} + ... + p_n \delta_{K^E_n}$ at maturity $T$ with given interest rate $r >0$ and mean $\mathbb{E}^{\mu}(X) = S_0e^{rT}$. 

The function $A$ is given for a finite number of strikes $K_1^A,...,K_m^A$ and interpolated linearly between them. Additionally, we know that $A(K) = K-S_0$ has to hold for (at least) all strikes $K \geq K_n^E e^{-rT}$, as we have by the definition of $\mu$ that the upper bound $E(e^{rT}K)$ coincides for these strikes with the lower bound given by $K-S_0$.
Note further that if we are given a set of strikes $K_1^A,...,K_m^A$ such that $A(K_i^A) > K_i^A - S_0$ for all $i=1,...,m$ then we can always add in a final strike $K_{m+1}$ with $A(K_{m+1}) = K_{m+1}-S_0$ using the extension method from (\ref{exa}) below. The new set of prices $A(K_1),...,A(K_{m+1})$ will then satisfy the conditions of Theorem~\ref{nco} again. Therefore we will assume from now on, without loss of generality, that the final strike, denoted by $K_m^A$ satisfies $A(K_m^A) = K_m^A-S_0$. 
We can then write the functions $A$ and $E$ as
\begin{align}\label{gae} 
A(K) &= \max\{0,s_1^A(K-S_d^1),...,s_{m-1}^A(K-S_d^{m-1}),K-S_0\} \notag \\
E(K) &= \max\{0,s_1^E(K-K_1^E)+d_1,...,s_{n-1}^E(K-K_{n-1}^E)+d_{n-1}, \\ 
         & \hspace{ 13 mm} e^{-rT}K-S_0\}, \notag
\end{align}
where the linear pieces are, without loss of generality, ordered by appearance. In Figure~\ref{fig:gen_setting} below the general setting is depicted, where the given European and American prices as functions of the strike $K$ are denoted by $E$ and $A$ respectively.
\begin{figure}[ht]
	\centering
	\includegraphics[width=\textwidth, trim=4cm 2cm 3cm 2cm, clip=true]{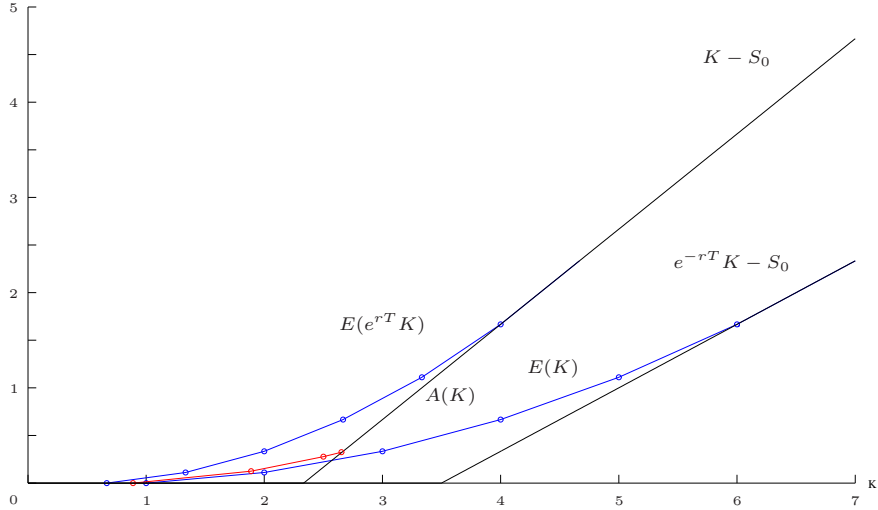}
	\caption{Given American and European prices with bounds}
	\label{fig:gen_setting}
\end{figure}

The idea now is to construct the process $(S_t)_{t\geq 0}$ by embedding one linear piece of $A$ after the other, where their order is determined by their {\it critical times}, which we will define below. After we embed a linear piece of $A$, we will split the initial picture $P$, given by $A$ and $E$, into two subpictures $P_1$ and $P_2$, where the new functions $A_i$ and $E_i$, $i=1,2$,  then again satisfy the conditions of Theorem~\ref{nco} and Lemma~\ref{coe} and due to the special choice of the critical times can be treated separately. As the number of linear pieces of $A$ that are left to embed in the subpictures $P_1$ and $P_2$ is reduced by one in each step and the European $E$ can be embedded at maturity $T$, we can argue inductively that the algorithm embeds $A$ and $E$ in finitely many steps.

\subsection{Algorithm}\label{alg}

In this section we outline the algorithm embedding the fuctions $A$ and $E$, where each step will be explained in more detail in the subsequent sections.
\begin{enumerate}
\item Set $t_{old}^{*} = 0$, $S_0 = \mathbb{E}(e^{-rT}S_T)$.
\item Extend A linearly beyond $\tilde{K} = \inf\{ K\geq 0: A(K) = K-S_0\}$ up to the first atom of $E$ where the necessary condition $$A'(K+)K-A(K) \geq E'(K+)K-E(K)$$ from Theorem~\ref{nco} is violated. From that strike on $A'(K+)$ is determined such that this condition is fulfilled with equality. Denote this extension by $\tilde{A}$ and the number of linear pieces of $\tilde{A}$ by $N_{\tilde{A}}$.
\item Compute the critical time $t_c^{*}$ and the critical strike $K^{*}$, determining the linear piece $s_k^A(K-S_d^{k})$ of $\tilde{A}$, where $k \in\{1,...,N_{\tilde{A}}\}$, that should be embedded next.
\item Embed $s_k^A(K-S_d^{k})$ by jumping the mass $p_d$ to $S_d$ and $p_u$ to $S_u$ at time $t^{*} =t_{old}^{*} +  t_c^{*} $, where $p_d = e^{rt_c^{*}}s_k^A$, $S_d = S_d^{k}$, $p_u = 1 - p_d$, $S_u = \frac{S_0e^{rt_c^{*}}-p_d S_d^k}{p_u}$. For $t_{old}^{*} < t < t^{*}$ set $S_t = e^{r(t-t_{old}^{*})}S_0$. Update $t_{old}^{*} = t^{*}$.
\item Split $\mu$ into $\mu_1$ and $\mu_2$, the given European prices $E$ into $E_1$ and $E_2$ and the given function $A$ into $A_1$ and $A_2$.
\item If $A_1 \neq E_1 \vee (K-S_d)_{+}$ set $A = A_1$, $E= E_1$ and $S_0 = S_d$ then go to $2.$, else embed $E_1$ at T.
\item If $A_2 \neq E_2 \vee (K-S_u)_{+}$ set $A = A_2$, $E=E_2$ and $S_0 = S_u$ then go to $2$., else embed $E_2$ at T.
\end{enumerate}


\subsection{Existence and calculation of the critical time}\label{ecct}

In this section we will construct a method to determine the critical time $t_c^{*}$, which will tell us when to embed the next linear piece of the given function $A$.
The actual jump of $S$ then occurs at $t^{*} = t_{old}^{*} + t_c^{*}$, where $t_{old}^{*}$ is the time where the parent node was embedded or $0$ in the first step. 

As we want to interpret the function $A$ for a fixed strike $K \geq 0$ as the American put option price on an unknown underlying price process $S$, we intend to split the function $A$ at $t^{*}$ into two independent functions $A_1$ and $A_2$ that can again be interpreted as American put option prices, where the underlying price process then starts at time $t^{*}$ in $S_d$ or $S_u$ respectively.

It follows that the contract length for the European Put price functions $E_1$ and $E_2$ has to be modified to $(T-t^{*})$. This directly affects the upper bound $\overline{A}$ given by 
\begin{align}\label{aub} 
\overline{A}(K,t) = E(e^{r((T-t_{old}^{*})-t)}K) 
\end{align}
for $0 \leq t \leq T-t_{old}^{*}$, which will play a crucial role in finding the critical time $t_c^{*}$.

Furthermore, we have the problem that $A$ only provides information on the underlying $S$ up to the strike $K_m^A$ above which exercising $A$ immediately is optimal. This information is not enough though to consider $A_1$ and $A_2$ independently forcing us to generate additional information on the underlying $S$ by extending $A$ beyond $K_m^A$. As long as this extension still satisfies the necessary conditions in Theorem~\ref{nco} this extension will not affect the American Put prices with respect to the underlying $S$, since $K-S_0$ will dominate these payoffs for $K \geq K_m^A$.

By extending $A$ linearly beyond $K_m^A$, only correcting the slope $A'(K+)$ when in an atom of $E$, where the condition $A'(K+)K-A(K) \geq E'(K+)K-E(K)$ is violated, we obtain
\begin{align}\label{exa} 
\tilde{A}(K) = \begin{cases}
A(K) &, 0 \leq K \leq K_m^A\\
s_{m-1}^A(K-S_d^{m-1}) &, K_m^A < K \leq K_p^E\\
\tilde{A}'(K_i^E+)(K-K_i^E) + \tilde{A}(K_i^E) &, K_i^E < K \leq K_{i+1}^E\\
\tilde{A}'(K_{N_E}^E+)(K-K_{N_E}^E) + \tilde{A}(K_{N_E}^E) &, K \geq K_{N_E}^E,
\end{cases}
\end{align}
where $i=p,...,N_E-1$, $\tilde{A}'(K_i^E+) =E'(K_i^E+)+\frac{\tilde{A}(K_i^E) - E(K_i^E)}{K_i^E}$ and $K_p^E$ is the first atom of $E$ after $K_m^A$ where the necessary condition is violated (Fig. \ref{fig:gen_ext}). Further set $N_{\tilde{A}} = N_A + N_E - p$. 

\begin{figure}[ht]
	\centering
	\includegraphics[width=\textwidth, trim=5cm 2cm 3cm 2cm, clip=true]{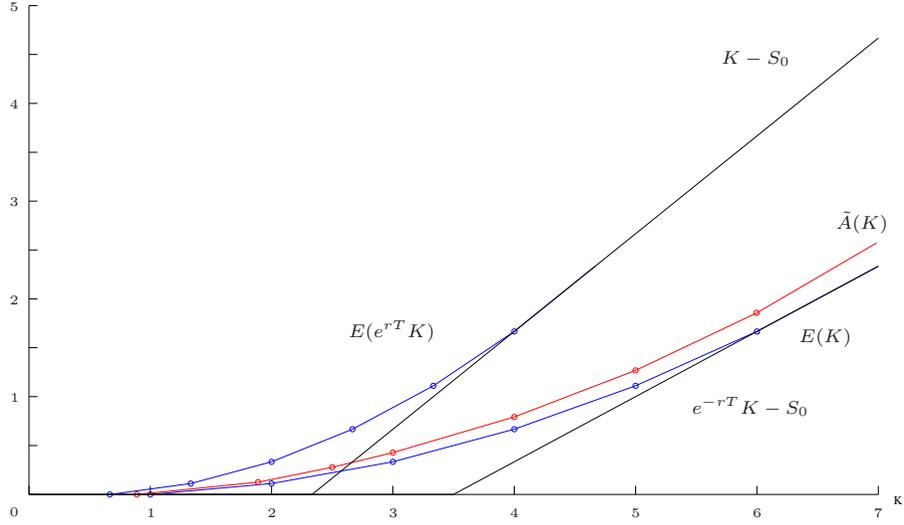}
	\caption{Extension of $A$}
	\label{fig:gen_ext}
\end{figure}

\begin{lemma}\label{eafc} 
Suppose the functions $A$ and $E$ given by (\ref{gae}) satisfy the conditions of Theorem~\ref{nco} and Lemma~\ref{coe}, then $A$ can be extended as in (\ref{exa}) to $\tilde{A}$, where $\tilde{A}$ and $E$ satisfy again the conditions of Theorem~\ref{nco}, except that $\tilde{A}$ no longer has $K-S_0$ as lower bound.
\end{lemma}

\begin{proof}
Let us start by pointing out that the condition $\tilde{A}'(K+)K-\tilde{A}(K) \geq E'(K+)K-E(K)$ is trivially fulfilled for all $K\geq 0$ by the choice of the extension $\tilde{A}$.

To see that $\tilde{A}$ is bounded below by $E$ remember that this is fulfilled up to $K_m^A$ by the assumptions on $A$ and $E$. Hence for $E$ to exceed $\tilde{A}$ between $K_m^A$ and $K_p^E$ we would need $E'(K+) > \tilde{A}'(K+)$ which can be ruled out, since we know that $\tilde{A}'(K+)K-\tilde{A}(K) \geq E'(K+)K-E(K)$ and $\tilde{A} \geq E$ holds in $K_m^A$. For $K \geq K_p^E$ we can argue inductively for each of the intervals, since the condition already has to hold in the respective left endpoint of the interval $K_i^E$, $i=p,...,n-1$.

Next we will show that $\tilde{A}(K)$ is bounded above by $\overline{A}(K,0)$, which is given in (\ref{aub}) above. Note that $\overline{A}(K,0) \geq A(K)$ for all $K\geq 0$. We will now show that we actually have $\tilde{A}(K) \leq A(K)$ for all $K \geq 0$. Up to $K_p^E$ this is trivially fulfilled by definition of $A$ and $\tilde{A}$. From $K_p^E$ onwards we have that $\tilde{A}'(K+)K - \tilde{A}(K) = E'(K+)K-E(K)$ and therefore $A'(K+)K-A(K) \geq \tilde{A}'(K+)K-\tilde{A}(K)$ has to hold for all $K \geq 0$ by the assumptions on $A$ and $E$ in Theorem~\ref{nco}. Using the fact that $A(K_p^E) \geq \tilde{A}(K_p^E)$ we can then conclude that we must have $A'(K_p^E+)\geq \tilde{A}'(K_p^E+)$. This allows us now to argue inductively and in the same way as for the lower bound to obtain that $\overline{A}(K,0) \geq \tilde{A}(K)$.

That $\tilde{A}$ is increasing for all $K \geq 0$ is an immediate consequence of the facts that $\tilde{A} \geq E$ and that $E' \geq 0$ as
\begin{align}\label{saex} 
\tilde{A}'(K_i+) =E'(K_i^E+)+\frac{\tilde{A}(K_i^E) - E(K_i^E)}{K_i^E} > 0
\end{align}
for $i \geq p$.

To prove that $\tilde{A}$ is convex it is enough to show that the slope of $\tilde{A}$ is increasing for any strike $K \geq K_p^E$, as we know already that $A$ is convex. Note that we can write
\begin{align*}
\tilde{A}(K_{i+1}^E) &= \tilde{A}'(K_i^E+)(K_{i+1}^E-K_i^E) + \tilde{A}(K_i^E)\\
E(K_{i+1}^E) &= E'(K_i^E+)(K_{i+1}^E-K_i^E) + E(K_i^E),
\end{align*}
since both $\tilde{A}$ and $E$ are piecewise linear functions. Further we have that $$E'(K_i^E+)K_i^E + E(K_i^E) = \tilde{A}'(K_i^E+)K_i^E + \tilde{A}(K_i^E)$$ for strikes $K_i^E \geq K_p^E$. Combined with the definition of the slope of $\tilde{A}$ from (\ref{saex}) we can then conclude that for $i \geq p$ we have
\begin{align*}
\tilde{A}'(K_{i+1}+) &=E'(K_{i+1}^E+)+\frac{\tilde{A}(K_{i+1}^E) - E(K_{i+1}^E)}{K_{i+1}^E}\\
			&=\tilde{A}'(K_i^E+) + (E'(K_{i+1}^E+) - E'(K_i^E+)),
\end{align*} 
which is increasing as $E'$ is and therefore $\tilde{A}$ has to be convex again.
\end{proof}

To determine a suitable critical time $t_c^{*}$, where the next linear piece of $A$ is embedded, we recall two important properties that we want to be fulfilled. First of all the underlying price process $S$ has to be a martingale and secondly we want the two subpictures, obtained by splitting at time $t^{*}$ in the {\it critical strike} $K^{*}$, to be independent. In this context we refer to the subpictures as being independent when no mass is exchanged between the two pictures $P_1$ and $P_2$ after $t^{*}$, allowing us to consider them separately.

We choose the critical time $t_c^{*}$ to be the first time $t$, where waiting any longer would result in $\overline{A}(K,t+\epsilon) < \tilde{A}(K)$ for some $K >0$ and any $\epsilon >0$, where $\overline{A}(K,t)$ denotes the upper bound on $\tilde{A}$ (see Fig.~\ref{fig:gen_crittime} below). We will show in Lemma~\ref{cct} that the critical time $t_c^{*}$ exists and is finite. Furthermore, we will see that the aforementioned properties are then satisfied.

\begin{figure}[ht]
	\centering
	\includegraphics[width=\textwidth, trim=5cm 2cm 1cm 2cm, clip=true]{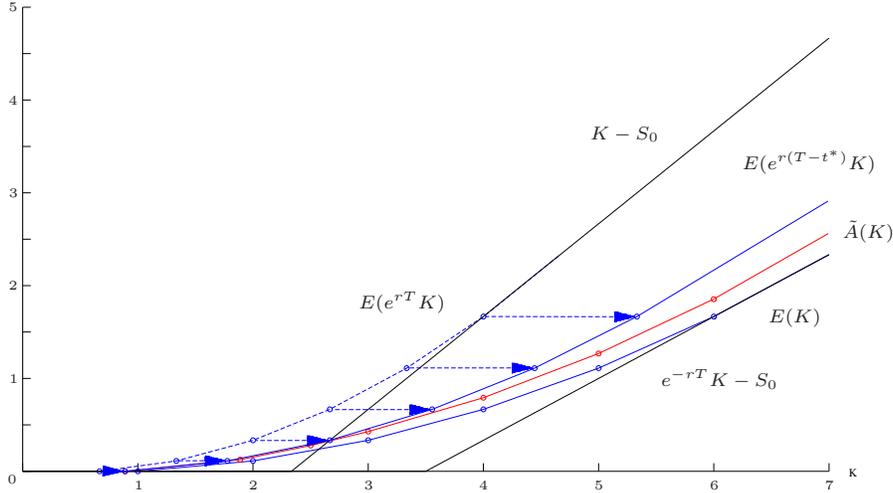}
	\caption{Critical time $t^{*}$ to embed the next piece of $A$}
	\label{fig:gen_crittime}
\end{figure}

Before we show the existence of the critical time $t_c^{*}$ note that the last linear piece of $A$, given by $K-S_0$, is already realised in the payoff as exercise at time $t_{old}^{*}$, since the the underlying process $S$ starts in $S_0$. Therefore the linear piece $K-S_0$ can be omitted when looking for the critical time. In the sequel we will refer to realising a payoff by jumping mass $p_d$ to $S_d$ (and $p_u$ to $S_u$) as embedding a linear piece of the American.

\begin{lemma}\label{cct} 
Suppose the given functions $A$ and $E$ satisfy the necessary conditions of Theorem~\ref{nco} and Lemma~\ref{coe}, where $A$ is extended to $\tilde{A}$ as in (\ref{exa}) and the European Put price function $E$ with contract length $T-t_{old}^{*}$ is given by the marginal distribution $\mu= p_1 \delta_{K^E_1} + ... + p_n \delta_{K^E_n}$ with maturity $T$.
Assume also that the upper bound $\overline{A}$ is given by $\overline{A}(K,t) =E(e^{r(T-t_{old}^{*}-t)}K)$.

Then we have that the critical time $t_c^{*}$ exists, is attained in an atom of the upper bound $\overline{A}$ and can be written as
\begin{align}\label{crt} 
t_c^{*} &= \inf_{i,j}  t^{i,j} \\
	  &= \inf_{i,j} \inf \{0 \leq t \leq T-t_{old}^{*}: \overline{A}(u_i,t) < f_j(u_i)\}, \notag
\end{align}
where $u_i = K_i^E e^{-r(T-t_{old}^{*}-t)}$, $i \in\{1,...,n\}$, $j \in\{1,...N_{\tilde{A}}\}$  and $f_j$ the $j$-th linear piece of $\tilde{A}$.
\end{lemma}

\begin{proof}
That the critical time $t_c^{*}$, if it exists, occurs in an atom of the upper bound is a simple consequence of convexity of the functions $\overline{A}$ and $\tilde{A}$. Hence the critical time $t_c^{*}$, should it exist, is given by (\ref{crt}).

Since we know that $\overline{A}(K,0) \geq \tilde{A}(K) \geq E(K)$ and
\begin{align*} 
\overline{A}(K,t) = E(e^{r((T-t_{old}^{*})-t)}K) \rightarrow E(K)
\end{align*}
for all $K\geq 0$, as $t \rightarrow (T-t_{old}^{*})$, the representation in (\ref{crt}) guarantees the existence of an $i$ and $j$ such that $t^{i,j} < \infty$ given that $\tilde{A}$ has not been embedded completely yet.
\end{proof}

Hence we define the critical strike $K^{*}$, where we will split the picture, to be given by the time-$t_c^{*}$ value of the smallest atom of $\overline{A}$ which intersects at the critical time $t_c^{*}$ with $\tilde{A}$, i.e.
\begin{align}
K^{*} = \inf \{K \geq 0 : \overline{A}(K,t_c^{*}) = \tilde{A}(K)\}.
\end{align}

The following lemma will give us now a simple way of determining $\inf_i t^{i,j}$, where $i \in \{1,...,N_E\}$ and $j \in \{1,...,N_{\tilde{A}}\}$, thereby highlighting the close connection between the necessary condition
\begin{align}
\tilde{A}'(K+)K-\tilde{A}(K) \geq E'(K+)K-E(K)
\end{align}
and the embedding time $t^{*}$ for a fixed linear piece $f_j$ of $\tilde{A}$.

\begin{prop}\label{cta} 
Suppose the given functions $A$ and $E$ satisfy the necessary conditions of Theorem~\ref{nco} and Lemma~\ref{coe}, where $A$ is extended to $\tilde{A}$ as in (\ref{exa}) and the European Put price function $E$ with contract length $T-t_{old}^{*}$ is given by the marginal distribution $\mu= p_1 \delta_{K^E_1} + ... + p_n \delta_{K^E_n}$ with maturity $T$.

Furthermore, assume, without loss of generality, that the linear piece of $\tilde{A}$ for which we want to find the critical time is given by $f_j(K)=s_j^A K-d$ and that the European $E$ coincides in $[K_i^E,K_{i+1}^E]$ with $g_i(K) = s_i^E K - d_i$ then we have
\begin{align}
\inf_i t^{i,j} = \frac{1}{r}\ln \left(\frac{s_{i^{*}}^E}{s_j^A}+\frac{d-d_{i^{*}}}{s_j^A}\frac{1}{K_{i^{*}}^E} \right) + (T-t_{old}^{*}),
\end{align}
where $i^{*} = \min\{1 \leq i \leq n : f_j'(K_i^E+)K_i^E-f_j(K_i^E) < g_i'(K_i^E+)K_i^E-g_i(K_i^E)\}$ or equivalently $i^{*} = \min \{1 \leq i \leq n : d < d_i\}$.
\end{prop}

\begin{proof}
The $i$-th linear piece $g_{t,i}(K)$ of the upper bound $\overline{A}(e^{-r((T-t_{old}^{*})-t)}K,t)$ at time $t$ in the interval $[K_i^E e^{-r((T-t_{old}^{*})-t)},K_{i+1}^E e^{-r((T-t_{old}^{*})-t)}]$ is given by $g_{t,i}(K) = g_i(e^{r((T-t_{old}^{*})-t)}K)$. To determine for any fixed linear piece $g_{t,i}(K)$, $i \in\{1,...,n\}$ the time $t^{i,j}$ when the atom $K_i^E e^{-r((T-t_{old}^{*})-t)}$ intersects with $f_j$ we rewrite $g_{t,i}(K)$ as follows
\begin{align*}
g_{t,i}(K)=s_i^E e^{r(T-t_{old}^{*})} K - d_i - s_i^E(e^{r(T-t_{old}^{*})}-e^{r((T-t_{old}^{*})-t)})K,
\end{align*}
which can then be interpreted as a clockwise rotation about the fixed point $(0,-d_i)$, as $t$ increases. In addition we know the strike $\hat{K}$ where the atom $K_i^E e^{-r((T-t_{old}^{*})-t)}$ has to hit $f_j$, since the value of $g_{t,i}(K)$ remains unchanged in the atom $K_i^E e^{-r((T-t_{old}^{*})-t)}$ over time, as we have $\overline{A}(K,t) = E(e^{r((T-t_{old}^{*})-t)}K)$. This allows us to obtain the candidate time
\begin{align*}
t^{i,j}= \frac{1}{r} \ln \left(\frac{s_i^E}{s_j^A}+\frac{d-d_i}{s_j^A} \frac{1}{K_i^E}\right) + (T-t_{old}^{*})
\end{align*}
by setting $g_{t,i}(\hat{K}) = g_i(K_i^E)$ and solving for $t$. 

This result now tells us that $t^{i,j}$ is a decreasing function of $K_i^E$ for $d > d_i$ as $s_i^E$,$s_j^A$, $r$ and $T$ are all positive constants, implying that for the two consecutive atoms $K_i^E$ and $K_{i+1}^E$, lying on the same linear piece $g_i$, the right atom $K_{i+1}^E$ will give a smaller candidate time. As $K_{i+1}^E$ is also the left-side endpoint of the next linear piece we can conclude by induction that as long as a linear piece $g_k$, $k\geq i$, still satisfies $d > d_k$ its right-side endpoint will attain a smaller candidate time than any atom before.

Analogously we see that for $d < d_i$ the function $t^{i,j}$ is increasing in $K_i^E$. Hence the critical time has to be attained in the atom $K_{i^{*}}^E$, which is the rightmost atom still lying on a linear piece $g_k$ satisfying $d \geq d_k$, but at the same time is the first atom lying on a linear piece $g_{k+1}$ where $d < d_{k+1}$. The existence of this atom $K_{i^{*}}^E$ is guaranteed by the fact that $d_n = S_0$, whereas $d < S_0$ for any linear piece of $A$ that is not embedded yet.
\end{proof}

\begin{remark}\label{roct} 
\begin{enumerate}[(i)]
\item This result implies that the critical time for a fixed linear piece $f_i$ of $A$ is attained in the  kink of $\overline{A}$ which corresponds to the European strike at which the Legendre-Fenchel condition between $f_i$ and $E$ is violated for the first time. Note that this is not a contradiction to the Legendre-Fenchel condition of Theorem~\ref{nco}, it simply means that the kink of $\overline{A}$ will not hit the linear piece $f_i$ in the interval where $A=f_i$, but to the right of it.

\item As it is possible that the upper bound $\overline{A}$ intersects with $\tilde{A}$ at the critical time $t_c^{*}$ in a kink of $\tilde{A}$ we need to specify which of the two linear pieces of $\tilde{A}$ we will embed. Proposition~\ref{cta} tells us now that we have to take the left linear piece given by $\tilde{A}'(K-)(K-K^{*})+\tilde{A}(K^{*})$.
\end{enumerate}
\end{remark}

\subsection{The splitting procedure}

After we determined the embedding time $t^{*} = t_{old}^{*} + t_c^{*}$ and the critical strike $K^{*}$ we will divide the functions $A$ and $E$ into two separate parts $A_{1}$, $A_{2}$, and $E_{1}$, $E_{2}$ respectively, such that $A_i$, $E_i$ $i \in \{1,2\}$ satisfy again all the conditions in Theorem~\ref{nco} and from which it will be possible to recover the initial functions $A$ and $E$.

\subsubsection{Splitting of the European Put option prices $E$}

To obtain $E_{1}$ and $E_{2}$ from $E$ we have to split $\mu$, the marginal distribution given at maturity $T$. Since the critical strike $K^{*}$ is given in time-$t^{*}$ value, the respective atom of $\mu$, where we have to split, is $K^{*}e^{r(T-t^{*})}$. The following lemma will show how to split $\mu$ into $\mu_1$ and $\mu_2$ and how to recover $E$ from $E_1$ and $E_2$.

\begin{prop}\label{eus} 
Assume the given functions $A$ and $E$ satisfy the necessary conditions of Theorem~\ref{nco} and Lemma~\ref{coe}, where $A$ is extended to $\tilde{A}$ as in (\ref{exa}) and the European Put price function $E$ with contract length $T-t_{old}^{*}$ is given by the marginal distribution $\mu= p_1 \delta_{K^E_1} + ... + p_n \delta_{K^E_n}$ with mean $\mathbb{E}^{\mu}(X)=e^{r(T-t_{old}^{*})}S_0$ at maturity $T$. Suppose further that the time of the next jump $t^{*}$, the critical time $t_c^{*}$ and the associated critical strike $K^{*}$ were determined as in Section~\ref{ecct} at which point the linear piece $s_k^A(K-S_d^k)$ of $\tilde{A}$ is embedded by jumping the mass $p_d$ to $S_d$ and $p_u$ to $S_u$, where $p_d$, $S_d$, $p_u$ and $S_u$ are given in Section~\ref{alg}. For the time between the jumps set the underlying price process $S_t = \mathbb{E}^{\mu}(X)e^{-r(T-t)}$, where $t_{old}^{*} < t < t^{*}$.

Then we can write $\mu = p_d \mu_1 + p_u \mu_2$, where $\mu_1$ and $\mu_2$ are given by
\begin{align}\label{mu1}
\mu_1 = p_d^{-1}\left[\mu_{\big| [K^E_1,K^{*}e^{r(T-t^{*})})} + (p_d - \mathbb{P}(S_T < K^{*}e^{r(T-t^{*})})) \delta_{K^{*}e^{r(T-t^{*})}}\right]
\end{align}
and
\begin{align}\label{mu2}
\mu_2 = p_u^{-1}\left[(\mathbb{P}(S_T \leq K^{*}e^{r(T-t^{*})})- p_d) \delta_{K^{*}e^{r(T-t^{*})}} + \mu_{\big| (K^{*}e^{r(T-t^{*})},K^E_n]}\right],
\end{align}
and satisfy $\mathbb{E}^{\mu_1}(e^{-r(T-t^{*})}X_1)=S_d $ and $\mathbb{E}^{\mu_2}(e^{-r(T-t^{*})}X_2)=S_u$. Dividing the distribution $\mu$ into $\mu_1$ and $\mu_2$ the European Put option $E$ with maturity $T$ can be written as 
\begin{align}\label{sef} 
E^{\mu}(K) = e^{-rt_c^{*}}\left[p_d E_1^{\mu_1}(K) + p_u E_2^{\mu_2}(K)\right], 
\end{align}
where $E_1$ and $E_2$ are European Put options starting at $t^{*}$, having maturity $T$ and satisfying the conditions in Lemma~\ref{coe}.
\end{prop}

\begin{proof}
Firstly, let us show that the mass that is placed in $K^{*}e^{r(T-t^{*})}$ for either of the two distributions $\mu_1$ and $\mu_2$ is positive. Without loss of generality we can assume that $K^{*}e^{r(T-t^{*})}$ is the $l$-th atom of $\mu$ and that the linear piece we just embedded was $f_k = s_k^A(K-S_d^k)$. If we set $\hat{K}_1 = \max\{K_{l-1}^E,K_k^A\}$, where $K_k^A = \inf\{K \geq 0: \tilde{A}(K) = s_k^A(K-S_d^k)\}$ then we have for the upper bound $\overline{A}$, which is given by $\overline{A}(K,t) = E(e^{r((T-t_{old}^{*})-t)}K)$ that 
$$\overline{A}(K^{*},t_c^{*}) = \overline{A}'(K^{*}-,t_c^{*})(K^{*}-\hat{K}_1)+\overline{A}(\hat{K}_1,t_c^{*})$$ and at the same time for the extended American $\tilde{A}$ from (\ref{exa}) that $\tilde{A}(K^{*}) = \tilde{A}'(K^{*}-)(K^{*}-\hat{K}_1)+\tilde{A}(\hat{K}_1)$. By the definition of $t_c^{*}$ we see that $\overline{A}(K^{*},t_c^{*}) = \tilde{A}(K^{*})$. Combining this with the fact that we must have $\overline{A}(\hat{K}_1,t_c^{*}) > \tilde{A}(\hat{K}_1)$ at the critical time $t_c^{*}$ we can conclude that $\tilde{A}'(K^{*}-) > \overline{A}'(K^{*}-,t_c^{*})$. Then again, we can use that $p_d = e^{rt_c^{*}}\tilde{A}'(K^{*}-)$ and that
\begin{align*}
\overline{A}'(K^{*}-,t_c^{*}) = E'(e^{r((T-t_{old}^{*})-t_c^{*})}K^{*}-)e^{r((T-t_{old}^{*})-t_c^{*})}
\end{align*}
to see that $p_d > E'(e^{r((T-t_{old}^{*})-t_c^{*})}K^{*}-) e^{r((T-t_{old}^{*}))}= \mathbb{P}(S_T < K^{*}e^{r(T-t^{*})})$. To show the other inequality we set $\hat{K}_2 = \min\{K_{l+1}^E,K_{k+1}^A\}$, where $K_{k+1}^A= \sup\{K \geq 0: \tilde{A}(K) = s_k^A(K-S_d^k)\}$, and note that $\tilde{A}'(K+) \geq \tilde{A}'(K-)$ as $\tilde{A}$ is convex. We can then argue analogously to above that $p_d \leq \mathbb{P}(S_T \leq K^{*}e^{r(T-t^{*})})$, where the inequality turns around as we have now $K^{*} \leq \hat{K}_2$.

By the martingale property of $(S_t)_{t_{old}^{*} \leq t \leq t^{*}}$ we have
\begin{align}\label{lhs} 
\mathbb{E}^{\mu}(X) = S_0 e^{r(T-t_{old}^{*})} = S_0 e^{rt_c^{*}}e^{r(T-t^{*})}=(p_d S_d + p_u S_u) e^{r(T-t^{*})}.
\end{align}
At the same time we can write
\begin{align}\label{rhs} 
\mathbb{E}^{\mu}(X) = p_d \mathbb{E}^{\mu_1}(X_1) + p_u \mathbb{E}^{\mu_2}(X_2),
\end{align}
since we clearly have $\mu=p_d \mu_1 + p_u \mu_2$. Equating now (\ref{lhs}) and (\ref{rhs}) we obtain $\mathbb{E}^{\mu_1}(X_1)=S_d e^{r(T-t^{*})}$ and
$\mathbb{E}^{\mu_2}(X_2)=S_u e^{r(T-t^{*})}$.

We can then conclude that $E^{\mu}(K) = e^{-rt_c^{*}}\left[p_d E_1^{\mu_1}(K) + p_u E_2^{\mu_2}(K)\right]$, as we know that $\mu = p_d \mu_1 + p_u \mu_2$ and that $E_1$ and $E_2$ have contract length $(T-t^{*})$.
From the last two statements and Lemma \ref{blr} it follows now directly that $E_1$ and $E_2$ satisfy the conditions of Lemma \ref{coe}.
\end{proof}

\subsubsection{Splitting of the American put option prices $A$}

In the case of the European Put option prices $E$ the existence of a 1-1 correspondence between $E$ and $\mu$ allows us to split the function $E$ by dividing $\mu$. For the American put option prices $A$ this 1-1 correspondence to the marginal distribution at a fixed deterministic time does not exist, since the time when it is optimal to exercise the option depends on the path of the underlying. We therefore need a different method to split $A$ that still allows us to recover the original function $A$ from the two new functions $A_1$ and $A_2$. The idea behind the specific choice of split in (\ref{ams}) is that we want to separate the already embedded immediate exercise from the continuation value in each step.

\begin{prop}\label{asl} 
Assume the given functions $A$ and $E$ satisfy the necessary conditions of Theorem~\ref{nco} and Lemma~\ref{coe}, where $A$ is extended to $\tilde{A}$ as in (\ref{exa}) and the European Put price function $E$ with contract length $T-t_{old}^{*}$ is given by the marginal distribution $\mu= p_1 \delta_{K^E_1} + ... + p_n \delta_{K^E_n}$ with mean $\mathbb{E}^{\mu}(X)=e^{r(T-t_{old}^{*})}S_0$ at maturity $T$. Suppose further that the time of the next jump $t^{*}$, the critical time $t_c^{*}$ and the associated critical strike $K^{*}$ were determined as in Section~\ref{ecct} at which point the linear piece $s_k^A(K-S_d^k)$ of $\tilde{A}$ is embedded by jumping the mass $p_d$ to $S_d$ and $p_u$ to $S_u$, where $p_d$, $S_d$, $p_u$ and $S_u$ are given in Section~\ref{alg}. For the time between the jumps the underlying price process is set to be $S_t = \mathbb{E}^{\mu}(X)e^{-r(T-t)}$, where $t_{old}^{*} < t < t^{*}$.

Then the function $A$ can be split into
\begin{align}
A_1(K) = e^{rt_c^{*}}p_d^{-1} \max \{0,f_1,f_2,...,f_k \}
\end{align}
and
\begin{align}\label{a2t} 
A_2(K) = e^{rt_c^{*}}p_u^{-1} \left[ \max \{f_k,f_{k+1},...,f_{N_{\tilde{A}}},e^{-rt_c^{*}}K-S_0 \} - f_k \right],
\end{align}
where $f_i=s_i^A(K-S_d^i)$, $i=1,...,N_{\tilde{A}}$, are the given piecewise linear functions of $\tilde{A}$ and
\begin{align}\label{ams} 
A(K) = \max \{K-S_0, e^{-rt_c^{*}}(p_d A_1(K) + p_u A_2(K)) \}.
\end{align}
The functions $A_1$ and $E_1$ as well as the functions $A_2$ and $E_2$ will then satisfy the necessary conditions of Theorem~\ref{nco} again.
\end{prop}

\begin{proof}
To see that (\ref{ams}) is satisfied we note that for $0 \leq K \leq K^{*}$ we have $A_2(K) = 0$ and therefore by the definition of $A$ we have $A(K) = \max \{K-S_0,e^{-rt_c^{*}}p_d A_1(K)\}$ in that interval. For $K \geq K^{*}$ we have $A_1(K) = e^{rt_c^{*}}p_d^{-1} f_k(K)$ and $A_2(K) = e^{rt_c^{*}}p_u^{-1}((\tilde{A}(K) \vee (e^{-rt_c^{*}}K-S_0))-f_k(K))$ and therefore
\begin{align}
A(K) &= \max \{K-S_0,\tilde{A} \vee (e^{-rt_c^{*}}K-S_0) \}\\
	 &= \max \{K-S_0, e^{-rt_c^{*}}(p_d A_1(K) + p_u A_2(K)) \},\notag
\end{align}
which holds true by the definition of $A$ and $\tilde{A}$ and the fact that $K-S_0$ dominates $e^{-rt_c^{*}}K-S_0$ for all $K \geq 0$.

We then have to check that the necessary conditions from Theorem~\ref{nco} are satisfied in the left hand-side picture $P_1$, where our new American is now $A_1$ and the new European is $E_1$. 
To see that $A_1$ has to be increasing, we can argue that the linear extension of an increasing function is again increasing and the multi constant does not change that. Also we have that $A_1$ has to be a convex function as it is the maximum over linear functions multiplied by a positive constant.

Let us now show that $A_1$ and $E_1$ satisfy
\begin{align}\label{ncd} 
A'_1(K+)K - A_1(K) \geq E'_1(K+)K - E_1(K)
\end{align}
for all $K \geq 0$. In the case where $K \leq K^{*}$ we have $A_1(K) = e^{rt_c^{*}}p_d^{-1}A(K)$ and $E_1(K) = e^{rt_c^{*}}p_d^{-1}E(K)$. Since the original functions $A$ and $E$ satisfy this condition and $e^{rt_c^{*}}p_d^{-1} > 0$ the functions $A_1$ and $E_1$ inherit this property. 

Next we show that the condition (\ref{ncd}) also holds for $K \geq K^{*}$. From Lemma~\ref{app1} in the appendix we know that it is enough to check the condition for the atoms of $E_1$. Starting out with the last atom $K^{*}e^{r(T-t^{*})}$ of $E_1$ we have
\begin{align*}
A'_1(K^{*}e^{r(T-t^{*})}+)K^{*}e^{r(T-t^{*})} - A_1(K^{*}e^{r(T-t^{*})}) &= A'_1(K^{*}+)K^{*} - A_1(K^{*})\\
&= K^{*} - A_1(K^{*}),
\end{align*}
where the fact that $A_1$ is linearly extended beyond $K^{*}$ gives the first equality and the fact that $A'_1(K^{*}+) = e^{rt_c^{*}}p_d^{-1} f'_k =1$ gives the second equality. Then again, since $A_1(K^{*}) = E_1(K^{*}e^{r(T-t^{*})})$ and $E'_1(K^{*}e^{r(T-t^{*})}+) = e^{-r(T-t^{*})}$ we see that
\begin{align*}
K^{*} - A_1(K^{*}) = E'_1(K^{*}e^{r(T-t^{*})}+)K^{*}e^{r(T-t^{*})} - E_1(K^{*}e^{r(T-t^{*})}).
\end{align*}
This shows that the condition is fulfilled for $K=K^{*}e^{r(T-t^{*})}$, but since $E_1$ is a convex function it follows that $E'_1(K+)K - E_1(K)$ is increasing in $K$, which readily implies for $K^{*} \leq K \leq K^{*}e^{r(T-t^{*})}$
\begin{align*}
E'_1(K+)K - E_1(K) 	&\leq E'_1(K^{*}e^{r(T-t^{*})}+)K^{*}e^{r(T-t^{*})} - E_1(K^{*}e^{r(T-t^{*})}) \\
					&\leq A'_1(K^{*}e^{r(T-t^{*})})+)K^{*}e^{r(T-t^{*})} - A_1(K^{*}e^{r(T-t^{*})}) \\
					&=A'_1(K+)K - A_1(K),
\end{align*}
where the last equality is due to the fact that $A_1$ is linear beyond $K^{*}$. Hence we must have  $A'_1(K+)K-A_1(K) \geq E'_1(K+)K-E_1(K)$ for all $K \geq 0$.

To see that $E_1$ is a lower bound on $A_1$ we use that $E_1(K)=e^{rt_c^{*}}p_d^{-1}E(K)$ and $A_1(K) = e^{rt_c^{*}}p_d^{-1} A(K)$ for $0 \leq K \leq K^{*}$. As $e^{rt_c^{*}}p_d^{-1}$ is positive and we have $A(K) \geq E(K)$ in the original picture we obtain $A_1(K) \geq E_1(K)$ for $0 \leq K \leq K^{*}$. For $K\geq K^{*}$ we know already that the condition $A'_1(K+)K-A_1(K) \geq E'_1(K+)K-E_1(K)$ has to hold. Combined with the fact that $A_1(K^{*}) \geq E_1(K^{*})$ we obtain that $A_1'(K^{*}+) \geq E_1'(K^{*}+)$, which then implies that we must have $A_1(K) \geq E_1(K)$ for as long as the slope of $E_1$ does not change. By induction on the atoms of $\mu_1$ to the right of $K^{*}$ we obtain that $E_1$ is a lower bound on $A_1$ for all strikes $K \geq 0$.

To show that $\overline{A}_1(K,t_c^{*}) = E_1(e^{r(T-t^{*})}K)$ is an upper bound on $A_1$ we distinguish the two cases $0 \leq K \leq K^{*}$ and $K \geq K^{*}$. In the first case we can use again that $E_1(K)=e^{rt_c^{*}}p_d^{-1}E(K)$ and that $A_1(K) = e^{rt_c^{*}}p_d^{-1} A(K)$, which then only has to be combined with $E(e^{r(T-t^{*})}K) \geq A(K)$ to obtain the result. The second case follows using the definition of the time $t^{*}$, where we have $E(e^{r(T-t^{*})}K) \geq A(K)$ and $E(e^{r(T-t^{*})}K^*) = A(K^*)$ implying $E_1(e^{r(T-t^{*})}K^*) = A_1(K^*)$. Since the last atom of $\mu_1$ is $K^{*}e^{r(T-t^{*})}$ we can conclude that $E'_1(K+) = e^{-r(T-t^{*})}$ for any $K \geq K^{*}e^{r(T-t^{*})}$. Then again $\overline{A}'_1(K,t_c^{*})  = E_1'(e^{r(T-t^{*})} K) e^{r(T-t^{*})}$, which is $1$ and therefore coincides with $A'_1(K)$ for $K \geq K^{*}$. Hence we showed that $\overline{A}_1(K,t_c^{*}) \geq A_1(K)$ for all strikes $K \geq 0$.

To be able to split the initial picture into the two subpictures $P_1$ and $P_2$ we have to show that the necessary conditions from Theorem~\ref{nco} also hold in the right hand-side picture $P_2$. To see that $A_2$ is an increasing function we note that $\max\{f_{k},...,f_{m-1}\}$ is increasing, since each linear piece $f_i$ with $i=1,...,m-1$ is increasing. Subtracting $f_k$, does not affect the monotonicity since we have $f'_k \leq f'_i$ for all $i=k,...,m-1$ as they are ordered by appearance. To obtain $A_2$ we only have to consider $e^{rt_c^{*}}p_u^{-1} \max\{\max\{0, f_{k+1}-f_k,...,f_{m+n-p}-f_k\},e^{-rt_c^{*}}K-S_0-f_k\}$, which is again increasing as the maximum over increasing functions. Further it follows immediately that $A_2$ has to be convex, since it is the maximum over linear functions multiplied by the positive constant $e^{rt_c^{*}}p_u^{-1}$. It only remains to show that the condition
\begin{align}\label{nc2} 
A'_2(K+)K - A_2(K) \geq E'_2(K+)K - E_2(K)
\end{align}
holds for all $K \geq 0$. For $0 \leq K \leq \min\{K^{*}e^{r(T-t^{*})},K_m^A\}$, where $K_m^A = \inf\{K\geq 0: A(K) = K-S_0\}$, the condition is trivially fulfilled, since the left hand-side is non-negative by the monotonicity and convexity of $A_2$ and $E_2$ is constantly $0$ there. Lemma~\ref{bacs} shows that the condition is also fulfilled for $K\geq K_m^A$ as $\hat{A}$ is the extension of $A_2$  from (\ref{exa}), which leaves the case $\min\{K^{*}e^{r(T-t^{*})},K_m^A\} < K \leq K_m^A$.
For $\min\{K^{*}e^{r(T-t^{*})},K_m^A\} < K \leq K_m^A$ we can write $A_2(K) = p_u^{-1} (e^{rt_c^{*}}A(K)-p_d A_1(K))$ and $E_2(K) = p_u^{-1} (e^{rt_c^{*}}E(K) - p_d E_1(K))$. Hence the condition (\ref{nc2}) simplifies to
\begin{multline}
e^{rt_c^{*}}(A'(K+)K-A(K)) - p_d(A'_1(K+)K-A_1(K)) \\
\geq e^{rt_c^{*}} (E'(K+)K-E(K)) - p_d (E'_1(K+)K-E_1(K)).
\end{multline}
Then again we know from the necessary conditions on $A$ and $E$ that $A'(K+)K - A(K) \geq E'(K+)K - E(K)$. Combining this with the fact that for $K \geq K^{*}e^{r(T-t^{*})}$ we have $A_1(K) = K-S_d$ and $E_1(K) = e^{-r(T-t_c^{*})}K-S_d$ we obtain $A'_1(K+)K - A_1(K) = E'_1(K+)K - E_1(K)=S_d$. Therefore the condition has to hold for all strikes $K\geq 0$.

We still have to show that $E_2$ is a lower bound on $A_2$. Consider first the case $0 \leq K \leq e^{r(T-t^{*})}K^{*}$, where we know that  $E_2(K) = 0$ as the support of $\mu_2$ begins in $e^{r(T-t^{*})}K^{*}$. Since $A_2$ is given as the maximum over finitely many linear functions and $0$ we can immediately conclude that we must have $A_2(K) \geq E_2(K)$ for all strikes $0 \leq K \leq e^{r(T-t^{*})}K^{*}$. In the case where $K \geq e^{r(T-t^{*})}K^{*}$ we know already that $A_2(e^{r(T-t^{*})}K^{*}) \geq E_2(e^{r(T-t^{*})}K^{*})$ and since we showed that $A'_2(K+)K-A_2(K) \geq E'_2(K+)K-E_2(K)$ has to hold for all $K\geq 0$ we can conclude that $A'_2(e^{r(T-t^{*})}K^{*}+) \geq E'_2(e^{r(T-t^{*})}K^{*}+)$. Hence we have $A_2(K) \geq E_2(K)$ for all strikes where the right hand-side derivative of $E_2$ remains unchanged. This way we can show by induction on the atoms of $E_2$ that we must have $A_2(K) \geq E_2(K)$ for all strikes $K \geq 0$.

Finally we are left with showing that $\overline{A}_2$, given by $\overline{A}_2(K) = E_2(e^{r(T-t^{*})}K)$ is an upper bound on $A_2$. As this is trivially fulfilled for $K < K^{*}e^{r(T-t^{*})}$ it is enough to consider $K \geq K^{*}e^{r(T-t^{*})}$. To this end we note that we must have $$\overline{A}_2(K,0) =  p_u^{-1}(e^{rt_c^{*}}\overline{A}(K,t_c^{*}) - p_d \overline{A}_1(K,0))$$ by the definition of $\overline{A}_2$ and the representation of $E$ by $E_1$ and $E_2$ in Proposition~\ref{eus}. We can then rewrite $\overline{A}_2 \geq A_2$ as
\begin{multline}\label{ubin} 
p_u^{-1}(e^{rt_c^{*}}\overline{A}(K,t_c^{*}) - p_d \overline{A}_1(K,0)) \geq \\
e^{rt_c^{*}}p_u^{-1}(\max\{f_k,f_{k+1},...,f_{N_{\tilde{A}}},e^{-rt_c^{*}}K-S_0\} - f_k).
\end{multline}
We can now use the fact that $\overline{A}_1(K) = E_1(e^{r(T-t^{*})}K)$ and since we have $K \geq K^{*}e^{r(T-t^{*})}$ we obtain further that $p_d \overline{A}_1(K) = K-S_d$, which equals exactly $e^{rt_c^{*}}f_k$. Hence the inequality in (\ref{ubin}) reduces to $$\overline{A}(K,t_c^{*}) \geq \max\{f_k,f_{k+1},...,f_{m+n-p},e^{-rt_c^{*}}K-S_0\}$$ or equivalently $\overline{A}(K,t_c^{*}) \geq \max\{\tilde{A},e^{-rt_c^{*}}K-S_0\}$, which has to hold as we know that $\overline{A}$ is an upper bound on $\tilde{A}$ with
\begin{align*}
\overline{A}(K,t_c^{*}) & = E(e^{r(T-t^{*})}K)\\
					& = \overline{A}(Ke^{-rt_c^{*}},0)\\
					& \geq Ke^{-rt_c^{*}} - S_0.
\end{align*}
where the last inequality is because $\overline{A}$ is initially an upper bound on $A$. Hence $\overline{A}_2$ is an upper bound on $A_2$ for all strikes $K \geq 0$.
\end{proof}

This result shows that the initial picture can be divided into the two subpictures $P_1$ and $P_2$, where each of these pictures satisfies again the necessary conditions of Theorem~\ref{nco}. Note that the splitting of the function $\tilde{A}$ as in (\ref{ams}) can be interpreted as separating the immediate exercise from the continuation value. The additional term $(e^{-rt_c^{*}}K-S_0)-f_k$ in $A_2$ represents the immediate payoff in $S_u$ at time $t^{*}$, since
\begin{align*}
e^{-rt_c^{*}}K-S_0-f_k &= e^{-rt_c^{*}}K-S_0- e^{-rt_c^{*}}p_d(K-S_d) \\
 &= e^{-rt_c^{*}}p_u(K-S_u),
\end{align*}
where the last equality is obtained by using the definition of $S_u$.

\subsection{Convergence of the Algorithm}

After having defined the splitting procedure we are now able to state the following proposition, which will then allow us to argue that the embedding algorithm only needs a finite number of steps to produce an admissible price process that has $A$ and $E$ as its American and European Put option prices respectively.

\begin{prop}\label{eeam} 
Assume the functions $A$ and $E$ are given by (\ref{gae}) and satisfy the necessary conditions of Theorem~\ref{nco} and Lemma~\ref{coe}, where $A$ is extended to $\tilde{A}$ as in (\ref{exa}) and the European Put price function $E$ is given by the marginal distribution $\mu= p_1 \delta_{K^E_1} + ... + p_n \delta_{K^E_n}$ at maturity $T$. Suppose that of the $N_{\tilde{A}}$ linear pieces of $\tilde{A}$ the linear pieces added to the American Put price function $A$ by (\ref{exa}) are given by $f_i$, $i=m,...,N_{\tilde{A}}$, then $f_i$, $i=m,...,N_{\tilde{A}}$ are all embedded together at maturity $T$.
\end{prop}

\begin{proof}
Let us assume, without loss of generality, that for the rightmost linear piece $f_{m-1}$ of the original American $A$ there exists at least one strike $K \in [K_{m-1}^A,K_m^A]$ for which $f'_{m-1}(K+)K-f_{m-1}(K) > E'(K+)K-E(K)$, otherwise consider the first linear piece of $A$ to the left of $f_{m-1}$ where this condition is satisfied with respect to the correct interval. This assumption ensures that $f_{m-1}$ is not embedded together with the pieces $f_i$, $i=m,...,m+n-p$ at maturity. By the definition of $\tilde{A}$ we have for $i=m,...,m+n-p-1$ that
\begin{align*}
f_i'(K_i^E+)K_i^E-f_i(K_i^E) = E'(K_i^E+)K_i^E-E(K_i^E),
\end{align*}
but
\begin{align*}
f_i'(K_{i+1}^E+)K_{i+1}^E-f_{i}(K_{i+1}^E) < E'(K_{i+1}^E+)K_{i+1}^E-E(K_{i+1}^E),
\end{align*}
where we have to consider the last linear piece of $A$ separately. Combined with Remark~\ref{roct} we can then conclude that the linear pieces $f_i$, $i=m-1,...,m+n-p-1$, of $\tilde{A}$ attain their critical time in the right-side endpoint $K_{i+1}^E$ of their respective interval, as it is the first European strike at which the Legendre-Fenchel condition does not hold anymore.

Then again we know from the definition of $\tilde{A}$ and $\overline{A}$ that the linear piece of $\overline{A}$ on which $K_i^E e^{-rT}$ and $K_{i+1}^E e^{-rT}$ lie will coincide with $f_i$ for any $i=m,...,m+n-p-1$ at its critical time $t_c^{*}$, as the two linear functions agree for the strikes $K=0$ and $K= K_{i+1}^E e^{-r(T-t_c^{*})}$. Hence the critical time attained in $K_{i+1}^E$ will coincide with the time obtained by $K_i^E$. The convexity of $\tilde{A}$ then guarantees that the linear piece $f_{i-1}$ will have a smaller critical time than $f_i$ for all $i=m,...,m+n-p-1$, as $\overline{A}(K_i^E,t)$ --- the kink in $\overline{A}$ responsible for the critical time of $f_i$ --- will hit $f_{i-1}$ before hitting $f_i$. Analogously, we obtain that the last linear piece of $A$ will be embedded the last. 

Suppose for now that we are embedding $f_{m-1}$ as first linear piece of $A$ then the American $A_2$ in $P_2$ has to coincide with the European $E_2$, as the strikes where the slopes change are $K=K_p^E,...,K_{m+n-p}^E$ for both functions and Lemma~\ref{bacs} from the appendix ensures that we have $$A'_2(K+)K-A_2(K) = E'_2(K+)K-E_2(K)$$ for all $K \geq 0$.

Finally we still have to rule out that embedding another linear piece $f_k$, $k < m-1$, first could cause us to embed $f_i$, $i > m-1$, before $f_{m-1}$. This can be achieved by using Lemma~\ref{bacs}, noting that the extension of $A_2$ is obtained by transforming $\tilde{A}$ as in (\ref{a2t}) omitting to take the maximum with $e^{-rt^{*}}K-S_0$. The extension will then again be convex by Proposition~\ref{asl} allowing us to conduct the same line of argument as above.
\end{proof}

This proposition allows us now to determine how the left and right hand-side subpictures have to look like after we embedded the last linear piece of the original function $A$ that did not coincide (partially) with a linear piece of $E$. By the definition of the algorithm in Section~\ref{alg} and Proposition~\ref{eeam} we know that only pieces of the original function $A$ are passed down to the left hand-side picture, since none of the linear pieces of the extension are embedded before maturity $T$. Hence we have $A_1 = E_1 \vee (K-S_d)_{+}$, where $E_1$ appears as a linear piece of $A$ could lie on $E$ in the original picture. 

Similarly we obtain $A_2 = E_2 \vee (K-S_u)_{+}$ for the right hand-side picture. The reason for this is that after having embedded the last linear piece of the original $A$ that did not lie on a linear piece of $E$, we are only left with the linear pieces added by step $2$ of  the algorithm in Section~\ref{alg} or a piece coinciding with a linear piece of $E$. Then again, Lemma~\ref{bacs} guarantees that all these linear pieces satisfy $A'(K+)K-A(K) = E'(K+)K-E(K)$ for any $K \geq 0$ and that the linear pieces of $A$ and $E$ change at the same strikes implying that they have to coincide. Therefore we have $A_2 = E_2 \vee (K-S_u)_{+}$.  

The following corollary to Proposition~\ref{cta} will then provide us with an upper bound on the number of steps necessary to embed the given functions $A$ and $E$.

\begin{corollary}\label{nsba} 
Suppose the given functions $A$ and $E$ satisfy the necessary conditions of Theorem~\ref{nco} and Lemma~\ref{coe} and that the number of linear pieces of $A$ is given by $N_A$, then the total number of steps necessary to embed $A$ and $E$ is bounded above by $2N_A + 1$.
\end{corollary}
\begin{proof}
Using Proposition~\ref{eeam} we know that after $N_A$ steps we finished embedding $A$ and are left with at most $N_A+1$ subpictures. As the extension of the American in each of these subpictures coincides with the European, which can be embedded at once at maturity $T$ we obtain at most an additional $N_A+1$ steps. Hence the whole algorithm has to terminate after at most $2 N_A +1$ steps.
\end{proof}

We are now able to state the major theorem of this paper, which will show that the conditions given in Theorem~\ref{nco} are not only necessary for the absence of arbitrage, but indeed sufficient.

\begin{thm}
Suppose we are given any finite number of American and European Put option prices with fixed maturity $T$, such that the functions $A$ and $E$, obtained by interpolating linearly between the given prices, satisfy the conditions given in Theorem~\ref{nco} and Lemma~\ref{coe}. Using the algorithm in Section~\ref{alg} a model $(\mathbb{Q},(S_t)_{t \geq 0})$ can be constructed such that $(e^{-rt}S_t)_{t \geq 0}$ is a martingale with $e^{-rT}\mathbb{E}^{\mathbb{Q}}(K-S_T)_+ = E(K)$ and $\sup_{0 \leq \tau \leq T} \mathbb{E}^{\mathbb{Q}}[e^{-r\tau}(K-S_{\tau})_+] = A(K)$.
\end{thm}

\begin{proof}
By construction, the underlying $S$ is a martingale, since in each step where we are embedding a linear piece of $A$ we choose the upper node $S_u$ by the martingale property and the process $S$ grows between the jumps at the interest rate. Further we know from Proposition~\ref{eus} that $\mathbb{E}^{\mu_1}(e^{-r(T-t^{*})}X)=S_d$ and $\mathbb{E}^{\mu_2}(e^{-r(T-t^{*})}X)=S_u$, guaranteeing that the martingale property is preserved in the last embedding step in each subpicture.

To see that the European Put option prices on the underlying $S$ coincide with the given prices $E$ we recall from Proposition~\ref{eus} that the sum of the marginal distributions at maturity $T$ in the subpictures coincides with the distribution implied by $E$ at maturity $T$.

Finally we still need to show that the American put option prices on the underlying $S$ agree with the given prices $A$. To this end we first show that it cannot be optimal to exercise between jumps. For a fixed path of the underlying $S$ we have for $t_1 < t < t_2$, where $t_1$ and $t_2$ are jump-times for this path, that $e^{-rt_1}K > e^{-rt}K$, and as $e^{-rt_1}S_{t_1} = e^{-rt}S_t$ we obtain $e^{-rt_1}(K-S_{t_1})_+ \geq e^{-rt}(K-S_t)_+$. Hence optimal exercise can only occur at the actual jump times $t_j$. Let us denote by $n_i(t_j)$ the possible asset prices that $S$ can assume at time $t_j$, where $i \in \{1,...,\# \{ \mbox{values $S$ can take at time $t_j$} \} \}$. We can then proceed by defining the height $h$ of a node $n_i(t_j)$ at time $t_j$  for $j \in \{0,...,m\}$, $t_m=T$, and $i \in \{1,...,\# \{ \mbox{nodes at time $t_j$} \} \}$ by
\begin{align*}
h(n_i(t_j)) = \begin{cases}
0 &, \mbox{ if } t_j = T\\
1 + \max_{k}\{h(c_k(i,t_j))\}  &, \mbox{ otherwise}
\end{cases}
\end{align*}
where $c_k(i,t_j)$ denotes the $k$-th direct child of the node $n_i(t_j)$.

If we can show now that the value of the American put option in each node $n_i(t_j)$  and for each strike $K$, denoted by
\begin{align}\label{osp} 
v(K,t_j,n_i(t_j)) = \sup_{t_j \leq \tau \leq T} \mathbb{E}[e^{-r\tau}(K-S_{\tau})_{+} | S_{t_j} = n_i(t_j)]
\end{align}
coincides with the price given by $A(K,t_j,n_i(t_j))$, which is  obtained by following the transformation of $A$ by the algorithm in Section~\ref{alg} up to the subpicture, where we just jumped to the node $n_i(t_j)$, then we have shown that 
\begin{align}
\sup_{0 \leq \tau \leq T} \mathbb{E}[e^{-r\tau}(K-S_{\tau})_+ ] = A(K)
\end{align}
has to hold. By the Dynamic Programming Principle (Theorem 21.7 in \citet{BJO}), the optimal stopping problem in (\ref{osp}) can be rewritten as the Bellman equation
\begin{align}\label{beq} 
v(K,t_j,n_i(t_j)) = \max \{ e^{-rt_j}(K-n_i(t_j))_+, \sum_{k=1}^{N_{c(i,t_j)}} p_k v(K,t_{j_k},n_k(t_{j_k}))\},
\end{align}
where $N_{c(i,t_j)}$ is the number of direct children of the node $n_i(t_j)$ and $p_k$ the probability of being at time $t_{j_k}$ in $n_k(t_{j_k})$. Note that by the construction of the algorithm the number of direct children $N_{c(i,t_j)}$ is $2$ for any node with height $h(n_i(t_j)) > 1$. For nodes with height $h(n_i(t_j)) = 1$ we can have more than $2$ direct children, as all remaining linear pieces have to be embedded. Using (\ref{beq}) we can now prove that 
\begin{align}\label{pcn} 
v(K,t_j,n_i(t_j)) = A(K,t_j,n_i(t_j))
\end{align}
by induction on the height of the nodes $n_i(t_j)$. For a node of height $1$ we know from step 6, or 7 resp., of the algorithm in Section~\ref{alg} that $A(K,t_j,n_i(t_j)) = (K-n_i(t_j))_+ \vee E(K)$, where $E$ is the European with contract length $(T-t_j)$ and marginal distribution given by the direct children of the node $n_i(t_j)$ and their transition probabilities. Hence the value of $E$ agrees with the second expression on the right hand-side of (\ref{beq}) and therefore we have that for nodes of height $1$ the equation in (\ref{pcn}) is satisfied.

Suppose now that we know $v(K,t_j,n_i(t_j)) =  A(K,t_j,n_i(t_j))$ for all nodes up to a height of $n$. Then again we must have $v(K,t_j,n_i(t_j)) =  A(K,t_j,n_i(t_j))$ for nodes $n_i(t_j)$ of height $(n+1)$, as the definition of the given prices for nodes of height $(n+1)$ in (\ref{ams}) is the maximum over the immediate exercise at that node $(K-n_i(t_j))_+$ and $p_1  v(K,t_{j_1},n_1(t_{j_1})) + p_2  v(K,t_{j_2},n_2(t_{j_2}))$, coinciding with the continuation value in the Bellman-equation, as each node of height $h(n_i(t_j)) \geq 2$ has by construction exactly $2$ direct children $n_1(t_{j_1})$ and $n_2(t_{j_2})$. Hence we conclude by induction that the American put option prices on the underlying $S$ have to coincide with the given prices $A$.
\end{proof}

\section{Conclusion}
In this paper we presented no-arbitrage conditions on American put option prices in a model-independent setting, where our only financial assumptions were that we can buy and sell both types of derivatives initially at the given prices, and that we can trade in the underlying frictionlessly at a discrete number of times. 

Any violation of the conditions of Theorem~\ref{nco} implies the existence of a simple arbitrage strategy. More importantly, we also showed that there always exists a model under which the discounted expected payoffs coincide with the given American and European prices whenever all the conditions are satisfied.


We believe that the results of this paper can be applied in many different ways. Market makers and speculators alike could use the conditions of Theorem~\ref{nco} to find misspecifications in the market prices. Simple trading strategies, provided in the proof of Theorem~\ref{nco}, can then be used to generate arbitrage. Furthermore the necessary conditions present a way of verifying the plausibility of prices obtained by numerical procedures or to extrapolate non-quoted prices from existing market data. Additionally, the results presented in this paper can be used to get an estimate for the model-risk associated with a particular position in the set of American options. 

Lastly we think that the results of this paper lead to the following interesting and unanswered questions. Are the conditions of Theorem~\ref{nco} also sufficient in a generalised setting where the American and European prices are given as continuous (and convex) functions? What conclusions can be made about the range of prices for portfolios consisting of long and short position in American put options with different strikes. Is it possible to say something about the exercise behaviour of the long positions with respect to the exercise behaviour of the short positions? What are conditions for the absence of model-independent arbitrage in a market trading American and European Put options, where European option prices are known for different maturity dates? How do the conditions on the option prices change if the underlying is allowed to pay dividends?

\appendix
\section{Appendix}

\begin{lemma}\label{cae} 
Suppose the given functions $A$ and $E$ satisfy the necessary conditions $(i)$, $(iii)$ and $(iv)$ of Theorem~\ref{nco} and Lemma~\ref{coe}, then the following conditions are all equivalent:
\begin{enumerate}[(i)]
\item $\forall K  \geq 0: \forall \epsilon > 0: $
\begin{align}\label{eblf} 
\frac{A(K+\epsilon)-A(K)}{\epsilon}K-A(K) \geq \frac{E(K+\epsilon)-E(K)}{\epsilon}K-E(K).
\end{align}

\item There exists an $\tilde{\epsilon}=\tilde{\epsilon}(K)$ such that (\ref{eblf}) holds for all positive $\epsilon$ less than $\tilde{\epsilon}$.

\item $\forall K \geq 0: A'(K+)K-A(K) \geq E'(K+)K-E(K).$
\end{enumerate} 
\end{lemma}

\begin{rmk} 
Any of the conditions in Lemma~\ref{cae} above implies that for traded strikes  $K_j^E \leq K_i^A \leq K_{j'}^E \leq K_{i'}^A$ the discretized version 
\begin{align*}\label{dsac} 
\frac{A(K_{i'}^A) - A(K_i^A)}{K_{i'}^A-K_i^A} K_i^A - A(K_i^A) \geq \frac{E(K_{j'}^E) - E(K_j^E)}{K_{j'}^E-K_j^E} K_{j}^E - E(K_j^E)
\end{align*}
has to hold. The market exhibits model-independent arbitrage whenever the condition 
is violated. This follows from the convexity of the function $A$.
\end{rmk} 

\begin{proof}[Proof of Lemma~\ref{cae}]
The implications $(i) \Rightarrow (ii)$ and $(ii) \Rightarrow (iii)$ are trivially fulfilled, since the set of $\epsilon$ for which we consider the inequality is in each case a subset of the set of $\epsilon$ from the statement above.

We then only have to show  $(iii) \Rightarrow (i)$ to prove equivalence between the 3 statements. 
Note further that it is enough to consider the case $K>0$, since for $K=0$ we have $A(K) = E(K) = 0$. If we suppose now that the condition $A'(K+)K-A(K) \geq E'(K+)K-E(K)$ holds for $K >0$ then we can show that $\frac{A(K)-E(K)}{K}$ has to be increasing on any compact interval $[a,b] \subset (0,\infty)$. To prove this we use Theorem 1 from \citet{MIV} implying that it is enough to show that $\frac{A(K)-E(K)}{K}$ is continuous on $[a,b]$ and that for all $K \in (a,b)$ the right sided derivative exists and is non negative. Since we know that $A$ and $E$ are convex functions on $(0,\infty)$ we know that their right sided derivatives exist and that $\frac{A(K)-E(K)}{K}$ is continuous on any subinterval $[a,b] \subset (0,\infty)$. Let us consider now the right side derivative of $\frac{A(K)-E(K)}{K}$ given by
\begin{align*}
\partial_+ \frac{A(K)-E(K)}{K} 	&= \lim_{\epsilon \downarrow 0} \frac{1}{\epsilon}\left(\frac{(A(K+\epsilon)-E(K+\epsilon))}{K+\epsilon} - \frac{(A(K)-E(K))}{K}\right)\\
								&= \frac{1}{K^2}(A'(K+)K-A(K))-(E'(K+)K-E(K)),
\end{align*}
which is non-negative as we have $A'(K+)K-A(K) \geq E'(K+)K-E(K)$. Hence $\frac{A(K)-E(K)}{K}$ is increasing and we can therefore write
\begin{align*}
\epsilon \frac{A(K)-E(K)}{K} &\leq \int_K^{K+\epsilon} \frac{A(u)-E(u)}{u} du\\
						&\leq \int_K^{K+\epsilon} (A'(u+)-E'(u+))du\\
						&=A(K+\epsilon)-E(K+\epsilon)-(A(K)-E(K)),
\end{align*}
where the integral in the second line is well defined as a convex function is differentiable almost everywhere. The inequality in the second line is obtained by the assumption $A'(K+)K-A(K) \geq E'(K+)K-E(K)$. We have therefore shown that $\frac{A(K+\epsilon)-A(K)}{\epsilon}K-A(K) \geq \frac{E(K+\epsilon)-E(K)}{\epsilon}K-E(K)$ has to hold for any $\epsilon >0$ and $K\geq 0$.
\end{proof}

\begin{lemma}\label{app1} 
Assume the piecewise linear functions $A$ and $E$ satisfy the necessary conditions $(i)$,$(iii)$ and $(iv)$ of Theorem~\ref{nco} and Lemma~\ref{coe}. Suppose further that their kinks are in $K_1^A,...,K_m^A$ and $K_1^E,...,K_n^E$ respectively. Then the condition $A'(K+)K-A(K) \geq E'(K+)K-E(K)$ holds for all strikes $K \geq 0$ if and only if it holds in the kinks of $E$.
\end{lemma}
\begin{proof}
We only have to show that it is enough to have the condition fulfilled in all strikes $K_i^E$, $i=1,...,m$, since the other implication is trivially fulfilled.
Suppose now the condition is fulfilled in $K_i^E$ and choose a strike $K \in [K_i^E,K^A]$, where $K^A = \min_{j=1,...,m}\{K_j^A:  K_i^E < K_j^A < K_{i+1}\}$. This way we have
\begin{align*}
A'(K+)K-A(K) &\geq A'(K_i^E+)K-A(K)\\
	&= A'(K_i^E+)K - A'(K_i^E+)(K-K_i^E)-A(K_i^E)\\
	&= A'(K_i^E+)K_i^E - A(K_i^E)\\
	&\geq E'(K_i^E+)K_i^E - E(K_i^E),
\end{align*}
where the first inequality holds since $A'(K^A+) \geq A'(K^A-)$. To obtain the equality in the second line we simply use the fact that for any $K$ in that interval we can write $A(K)=A'(K_i^E+)(K-K_i^E)+A(K_i^E)$ and for the last inequality that  the condition is known to hold in $K_i^E$. But then again we can rewrite
\begin{align*}
E'(K_i^E+)K_i^E - E(K_i^E) &= E'(K_i^E+)K - E'(K_i^E+)(K-K_i^E)-E(K_i^E)\\
		&=E'(K+)K - E(K).
\end{align*}
Hence we have $A'(K+)K-A(K) \geq E'(K+)K - E(K)$. This leaves us to show that for any strike $K \in (K^A,K_{i+1}^E)$ the condition is fulfilled, but we can use the same argument now, inductively on the kinks of $A$ between $K^A$ and $K_{i+1}^E$, we have that the condition has to hold for any strike $K \in [K_i^E, K_{i+1}^E)$. Since the strike $K_i^E$ was taken arbitrarily we know that the condition has to hold for all strikes $K \in [K_1^E,\infty)$. Then again for any strike prior to $K_1^E$ the condition is trivially fulfilled, since we know that $A$ is increasing and convex and therefore has to satisfy $A'(K+)K-A(K) \geq 0$.
\end{proof}

\begin{lemma}\label{bacs} 
Assume the functions $A$ and $E$ given by (\ref{gae}) satisfy the necessary conditions of Theorem~\ref{nco} and Lemma~\ref{coe}, where $A$ is extended to $\tilde{A}$ as in (\ref{exa}) and the European Put price function $E$ with contract length $(T-t_{old}^{*})$ is given by the marginal distribution $\mu= p_1 \delta_{K^E_1} + ... + p_n \delta_{K^E_n}$ with mean $\mathbb{E}^{\mu}(X)=e^{r(T-t_{old}^{*})}S_0$ at maturity $T$. Suppose further that the time of the next jump $t^{*}$ and the associated critical strike $K^{*}$ were determined as in Section~\ref{ecct} at which point the linear piece $s_k^A(K-S_d^k)$ of $\tilde{A}$, denoted by $f_k$, is embedded.

Suppose further that $E_2$ is given by Proposition~\ref{eus} and define $\hat{A}(K) = e^{rt_c^{*}}p_u^{-1} \left(\max \{f_m,...,f_{N_{\tilde{A}}} \}- f_k\right)$ then it follows that $$\hat{A}'(K+)K-\hat{A}(K) = E'_2(K+)K-E_2(K)$$ for $K \geq K_p^E$, where $K_p^E$ and $N_{\tilde{A}}$ are defined in  Section~\ref{alg}.
\end{lemma}

\begin{rmk}
The result of Lemma~\ref{bacs} shows that we obtain the extension in the right hand-side sub-picture $P_2$ by transforming the extension in the original picture $P$, as the functions $\hat{A}$ and $A_2$ coincide in $K_p^E$ and the Legendre-Fenchel condition is satisfied with equality for $K \geq K_p^E$.
\end{rmk}

\begin{proof}[Proof of Lemma~\ref{bacs}]
We know already from Lemma~\ref{app1} that it is enough to check the condition in the atoms of $E_2$, which by the definition of $E_2$ in Proposition~\ref{eus} coincide with the ones of $E$. Consider therefore $K_j^E$, where $j \in \{p,...,n\}$ and assume, without loss of generality, that the American $A$ to the right of $K_j^E$ is given by $f_i$ then we have
\begin{multline*}
\hat{A}'(K_j^E+)K_j^E-\hat{A}(K_j^E) = \\
e^{rt_c^{*}}p_u^{-1} [ (f'_j(K_j^E+)K_j^E - f_j(K_j^E)) - (f'_k(K_j^E+)K_j^E - f_k(K_j^E))].
\end{multline*}
Furthermore we can use Proposition~\ref{eus} to write
\begin{align*}
E_2(K_j^E) = p_u^{-1} (e^{rt_c^{*}}E(K_j^E)-p_d E_1(K_j^E)),
\end{align*}
since Proposition~\ref{cta} guarantees that we have $K_j^E \geq K^{*}e^{r(T-t^{*})}$. As we are only considering strikes where $A'(K+)K-A(K) = E'(K+)K-E(K)$ holds, we get that the equation $\hat{A}'(K_j^E+)K_j^E-\hat{A}(K_j^E) = E'_2(K_j^E+)K_j^E-E_2(K_j^E)$ reduces to $e^{rt_c^{*}}(f'_k(K_j^E+)K_j^E-f_k(K_j^E)) = p_d (E_1'(K_j^E+)K_j^E-E_1(K_j^E))$. This equality has to hold though, since we know that $f_k(K) = e^{-rt_c^{*}}p_d(K-S_d)$ and $E_1(K)= e^{-r(T-t^{*})}K-S_d$ for $K \geq K^{*}e^{r(T-t^{*})}$.
\end{proof}


\newpage
\bibliographystyle{plainnat}
\bibliography{mybibliography}

\end{document}